\pdfoutput=1 

\newif\ifconference
\conferencefalse

\ifconference
\documentclass[twoside,leqno,11pt]{article}
\usepackage{ltexpprt}
\usepackage{hyperref}
\else
\documentclass[11pt]{article}
\usepackage[hyperindex,breaklinks,bookmarks]{hyperref}
\usepackage{amsthm}
\fi
\usepackage[utf8]{inputenc}
\usepackage[english]{babel}
\usepackage[dvipsnames]{xcolor}
\usepackage{amsmath,amsfonts}
\usepackage{amssymb}
\usepackage[showonlyrefs]{mathtools}
\usepackage{moresize}
\usepackage[letterpaper, portrait, left=1in, right=1in, top=1in, bottom=1in]{geometry}
\usepackage[numbers,sectionbib,longnamesfirst]{natbib}
\hypersetup{bookmarksdepth=2}
\usepackage[ruled,vlined,linesnumbered]{algorithm2e}
\usepackage[shortlabels]{enumitem}
\usepackage[toc]{appendix}
\usepackage{cleveref}
\usepackage{autonum}
\usepackage{mathrsfs}
\usepackage{pbox}

\newcommand{\ignore}[1]{}

\SetKwFor{RepTimes}{repeat}{times}{end}

\definecolor{darkgreen}{rgb}{0.00, 0.43, 0.00}

\newcommand\mc{\mathcal}
\newcommand\Mm{\mathcal{M}}
\newcommand\Rr{\mathcal{R}}

\newcommand\Pp{\mathcal{P}}
\newcommand\Dd{\mathcal{D}}
\newcommand\Ss{\mathcal{S}}
\newcommand\Ii{\mathcal{I}}
\DeclareMathOperator{\Var}{Var}

\ifconference \else
\newtheorem{theorem}{Theorem}

\newtheorem{definition}[theorem]{Definition}
\newtheorem{lemma}[theorem]{Lemma}

\newtheorem{fact}[theorem]{Fact}
\newtheorem{claim}[theorem]{Claim}

\theoremstyle{remark}
\newtheorem*{remark}{Remark}
\fi

\newcommand\blfootnote[1]{%
  \begingroup
  \renewcommand\thefootnote{}\footnote{#1}%
  \addtocounter{footnote}{-1}%
  \endgroup
}

\let\oldparagraph\paragraph
\renewcommand{\paragraph}[1]{\oldparagraph{\it #1}}

\date{}
\ifconference
\title{\Large Massively Parallel Computation on Embedded Planar Graphs\thanks{The authors are part of BARC, Basic Algorithms Research Copenhagen, supported by the VILLUM Foundation grant 16582.}}
\else
\title{Massively Parallel Computation on Embedded Planar Graphs}
\fi

\ifconference
\author{Jacob Holm\thanks{Basic Algorithms Research Copenhagen, University of Copenhagen, \texttt{jaho@di.ku.dk}.} \and Jakub Tětek\thanks{Basic Algorithms Research Copenhagen, University of Copenhagen, \texttt{j.tetek@gmail.com}.}}
\else
\author{
\begin{minipage}{12em}
\centering 
Jacob Holm \\ \texttt{jaho@di.ku.dk}
\end{minipage}
\begin{minipage}{12em}
\centering
Jakub Tětek\\\texttt{j.tetek@gmail.com}\\
\end{minipage}
\vspace{1em}\\
	Basic Algorithms Research Copenhagen,\\ University of Copenhagen
}
\fi

\begin{document}
%\section*{TODO}
%\begin{itemize}
%%    \todobullet{Jakub} deal with the boundary cases in Algorithm 1.
%%    \todobullet{Jakub} Deal with the $\delta$'s
%%    \todobullet{Jakub} add something about piecewise-linear embeddings to the ``beyond straight-line embeddings" section
%%    \todobullet{Both} think about where an illustration would be useful; I think we want to have a few nice pictures in this paper
%%    \todobullet{Jakub} at the end, replace ``$r$-decomposition'' $\rightarrow$ ``$r$-division''
% %   \todobullet{Jakub} division vs. pseudodivision
%%    \todobullet{Both} is it OK that there are new definitions in Preliminaries
%\end{itemize}
% \thispagestyle{empty}
%\newpage
\maketitle

\ifconference\else
\blfootnote{The authors are part of BARC, Basic Algorithms Research Copenhagen, supported by the VILLUM Foundation grant 16582.}
\fi

\ifconference
\fancyfoot[R]{\scriptsize{Copyright \textcopyright\ 2023 by SIAM\\
Unauthorized reproduction of this article is prohibited}}
\else
\thispagestyle{empty}
\fi

\begin{abstract}
\ifconference \small\baselineskip=9pt \fi
Many of the classic graph problems cannot be solved in the Massively Parallel Computation setting (MPC) with strongly sublinear space per machine and $o(\log n)$ rounds, unless the 1-vs-2 cycles conjecture is false. This is true even on planar graphs. Such problems include, for example, counting connected components, bipartition, minimum spanning tree problem, (approximate) shortest paths, and (approximate) diameter/radius.

In this paper, we show a way to get around this limitation. Specifically, we show that if we have a ``nice'' (for example, straight-line) embedding of the input graph, all the mentioned problems can be solved with $O(n^{2/3+\epsilon})$ space per machine in $O(1)$ rounds.
In conjunction with existing algorithms for computing the Delaunay triangulation, our results imply an MPC algorithm for exact Euclidean minimum spanning thee (EMST) that uses $O(n^{2/3 + \epsilon})$ space per machine and finishes in $O(1)$ rounds. This is the first improvement over a straightforward use of the standard Borůvka's algorithm with the Dauleanay triangulation algorithm of Goodrich [SODA 1997] which results in $\Theta(\log n)$ rounds. This also partially negatively answers a question of Andoni, Nikolov, Onak, and Yaroslavtsev [STOC 2014], asking for lower bounds for exact EMST.

We extend our algorithms to work with embeddings consisting of curves that are not ``too squiggly" (as formalized by the total absolute curvature). We do this via a new lemma which we believe is of independent interest and could be used to parameterize other geometric problems by the total absolute curvature. We also state several open problems regarding massively parallel computation on planar graphs.
\end{abstract}

\newpage
\clearpage
\setcounter{page}{1}
\section{Introduction}

In massively parallel computation (MPC), many of the classic graph problems require $\Omega(\log n)$ rounds if we have $O(n^{1-\epsilon})$ space per machine, under the 1-vs-2-cycles conjecture~\cite{beame2017communication,Nanongkai2020}. This includes, for example, counting connected components, bipartition, minimum spanning tree problem, (approximate) shortest paths, and (approximate) diameter/radius. Moreover, the lower bounds hold even for very restricted graph classes, such as planar graphs with bounded degrees.

However, when dealing with a planar graph, we would usually have available an embedding of the graph in the plane, along with the graph itself --- after all, the graph is presumably planar exactly because it is embedded in the plane. This is the case, for example, for maps, (single-layer) circuit boards, computational geometry applications, and more. This leads to a natural question: can we use the embedding to get around these lower bounds?

In this paper, we answer this question positively. Specifically, we give algorithms for all the above-mentioned problems in $O(1)$ rounds if we have $\Omega(n^{2/3+\epsilon})$ space per machine. Throughout the paper, we assume a straight-line embedding. In \Cref{sec:beyond_straight_lines}, we show that our results also generalize to embeddings with edges as differentiable curves, with the complexity being parameterized by their total absolute curvature. We achieve this by introducing a new lemma which we believe may be useful for getting other geometric problems parameterized by the total absolute curvature of the input. Our algorithms also generalize to the case when the graph is not planar but has $\ll n^2$ crossings.

%\jacob{Mention that nobody has used $r$-divisions in an MPC setting before (if true).}

All our algorithms are based on using an $r$-division of the input graph. This is the first time that $r$-divisions are used in the context of MPC. We believe the reason this has not been done before is the following. If we have $\Omega(n)$ space per machine, then the whole graph fits onto one machine, and we can solve any problem in a single round. On the other hand, if we have $O(n^{1-\epsilon})$ space per machine, it is impossible to find an $r$-division in $o(\log n)$ rounds under the $1$-vs-$2$-cycles conjecture, because otherwise our algorithms contradict the conjecture. We get around this by assuming a plane embedding of the input graph, allowing us to find an $r$-division. To the best of our knowledge, this paper is also the first to notice that having a geometric embedding of a planar graph can be used to solve some graph problems more efficiently.

Our results imply an exact MPC algorithm for the Euclidean minimum spanning tree problem that uses
%of an embedded planar graph
$O(1)$ rounds and $O(n^{2/3+\epsilon})$ space per machine, for any $\epsilon > 0$. This partially answers a question of \citet{Andoni2014}, who asked if it is possible to find exact euclidean MST in this complexity. This is in contrast to their lower bound which shows that in $\ell_\infty$ norm and $\Theta(\log n)$ dimensions, this is not possible under the 1-vs-2-cycles conjecture.
%This is perhaps surprising, as it is known that in the $\ell_\infty$ norm, this is not possible under the 1-vs-2-cycles conjecture \cite{Andoni2014}. 

Our algorithm for Euclidean MST works as follows. We start by using one of the existing algorithms \cite{Goodrich1997,Kramer2020} to compute a Delaunay triangulation of the input set. This can be done in $O(1)$ rounds and $O(n^\epsilon)$ space per machine. It is well-known that the EMST of a point set is a subset of its Delaunay triangulation. To get the euclidean minimum spanning tree, we then may use our MST algorithm on the Delaunay triangulation that uses $O(1)$ rounds and $O(n^{2/3+\epsilon})$ space per machine. The Delaunay triangulation is an embedded planar graph with a straight-line embedding, and the assumptions of our algorithm are thus satisfied.
% Our algorithm requires that the input graph is a straight-line embedding (even though this assumption can be relaxed), which is satisfied for a Delaunay triangulation.
One can easily get an algorithm with $O(\log n)$ rounds complexity, by using Borůvka's MST algorithm in place of our more efficient algorithm. The only previously known algorithm that uses $o(\log n)$ rounds and $n^{1-\Omega(1)}$ space per machine is that of \citet{Andoni2014} which only gives an approximate answer.

As a starting point, our MPC algorithms use a modification of the algorithm of \citet{chazelle2011online} for finding an $r$-division of a triangulation in sublinear time.  As a side-note, we show in Appendix~\ref{sec:estimate_Lipschitz_parameters} that the algorithm of \citet{chazelle2011online} can also be used to estimate additive Lipschitz graph parameters. This includes properties such as maximum matching, maximum independent set, or minimum dominating set. 
This class of problems has been considered before by \citet{Newman2013}.
While \citet{Newman2013} shows how to estimate a more general class of problems, they assumed that the graph has degrees bounded by $d$. Our algorithm works in general planar graphs. Moreover, our algorithm is simple and has complexity $O(\sqrt{n}\log^{3/2} n/\epsilon^{2.5})$. The approach of \cite{Newman2013} when used in conjunction with the property-testing algorithm of \citet{Kumar2021} results in complexity $O(%d^{136}/\epsilon^{15154}+
d^{686}/\epsilon^{30654})$ for $d$ being the maximum degree, and is unlikely to be practical.

Separators and $r$-divisions have been used before in parallel algorithms for planar graphs in the PRAM model \cite{klein1993linear,subramanian1995efficient,gianinazzi2020parallel,traff2000simple,cohen1993efficient}. However, that line of work hits very different challenges than we do. The difficulties in this paper are $(1)$ it is unclear how to compute the $r$-division and $(2)$ even given an $r$-division, it is unclear how to best solve the problem since the amount of local memory is limited. The second problem is not present in the PRAM model, while the first has a very different solution in PRAM \cite{klein1993linear}.

\subsection{Previous work}

\subsubsection{Massively parallel computation}

Until recently, no MPC algorithms have been known for planar graphs in the setting with $n^{1-\Omega(1)}$ memory per machine that would have better performance than the best known algorithm for general graphs. Recently, two papers that focus on bounded-arboricity graphs (which includes planar graphs) appeared. Specifically, in \cite{Behnezhad2019} the authors show algorithms for the maximum matching and maximum independent set problems that use $O(n^{\epsilon})$ space for any $\epsilon>0$ and $O(\log^2 \log n)$ rounds in bounded-arboricity graphs. This has been later improved to $O(\log \log n)$ by \citet{Ghaffari2020}. Still, these algorithms only use the bound on the arboricity which can also be satisfied by graphs that are far from being planar. No algorithm specific to planar graphs was previously known in the setting with $n^{1-\Omega(1)}$ memory per machine.

The $1$-vs-$2$-cycles conjecture states that for any $\epsilon > 0$, one needs $\Omega(\log n)$ rounds to tell apart one cycle of length $n$ and two cycles of length $n/2$ when using $n^{1-\Omega(1)}$ space per machine. This has been proven for some special classes of algorithms; for an overview of these results, see \cite[Section~3.2]{Nanongkai2020}. This conjecture is known to be equivalent to many common combinatorial problems \cite{Nanongkai2020}. Specifically, any of the problems that we consider in \Cref{sec:mpc_algorithms} need $\Omega(\log n)$ rounds with strongly sublinear space if this conjecture holds. We show that for any $\epsilon > 0$, we can get an algorithm with $O(1)$ rounds and $O(n^{2/3+\epsilon})$ space per machine by using the embedding. Our algorithms thus give a conditional separation between the settings with and without an embedding in MPC. Algorithms with constant round complexity for these problems were previously only obtained in a stronger model called Adaptive Massively Parallel Computation \cite{Behnezhad2020}.

\subsubsection{Algorithms for geometric graphs}
While we are aware of combinatorial embeddings being used in the previous work (such as in \cite{kao1993towards}), we are not aware of any papers specific to the setting where we have a geometrically embedded planar graph as the input. %The reason for this is presumably that one can get an embedding in linear time and there is thus no need for such assumption.

It is NP-hard to get an embedding with the smallest possible number of crossings \cite{Garey1983}, and it thus may give us additional power if we assume that we are given an embedding with a small number of crossings. \citet{Eppstein2010} have shown that in a graph with $O(n/\log^{(c)})$ crossings for a constant $c$, they are able to find an $r$-division in linear time, thus allowing them to get linear time algorithms for several problems, including the single source shortest paths problem.

%\citet{Eppstein2008} show how to use some other geometric properties of the road network than being almost planar to get efficient algorithms. We, however, do not want to make such further assumptions. There have been many papers on more restricted graph families that are geometric in nature, such as Euclidean graphs, geometric intersection graphs, and others, but these are out of the scope of this paper.

\subsubsection{PRAM algorithms on planar graphs}
Separators and $r$-divisions have been used previously in the context of PRAM, see e.g.\ \cite{klein1993linear,subramanian1995efficient,gianinazzi2020parallel,traff2000simple,cohen1993efficient}. 
\emph{Combinatorial} embeddings have been used before in algorithms for planar graphs, see e.g.\ \cite{thorup2004compact,kao1993towards,atallah1997efficient}. Moreover, an embedding may be computed efficiently even in PRAM \cite{klein1988efficient}, making it often unnecessary to assume that the embedding is part of the input both in sequential and PRAM computation. As we said above, we show that this is not the case in MPC. We are unaware of previous use of \emph{geometric} embeddings in algorithms.

\subsection{Techniques}
As a starting point, we use the algorithm of \citet{chazelle2011online} for finding an $r$-division of a triangulation in sublinear time. We define a notion of a hybrid $r$-division and show that the algorithm from \cite{chazelle2011online} can be modified to find a hybrid $r$-division in general planar graphs. In MPC, it is then possible to turn the hybrid $r$-division into a non-hybrid one. This can in turn be used to give algorithms for the above-mentioned problems. One obstacle to achieving the $O(n^{2/3+\epsilon})$ per machine space complexity is that in this space complexity, we are only able to find an $r$-division with boundary per region of $O(n^{1-\epsilon'})$ where $\epsilon'$ depends on $\epsilon$. For all the above-mentioned problems we are able to get around this rather large boundary size by giving a recursive algorithm.

Our algorithms in MPC rely on the inner workings of the sublinear-time $r$-division algorithm. We thus start by describing our modification of the algorithm from \cite{chazelle2011online}. We then describe below how this can be used in MPC.

\subsubsection{Computing a hybrid $r$-division in sublinear time}

%\begin{figure}[h]
%\centering
%\includegraphics[width=\textwidth]{picture}
%\caption{Illustration of our algorithm for computing a division oracle in sublinear time. Note that the sector graph is not necessarily equal to the dual graph (see \Cref{def:sector_graph}). Note also that the input and the sector graph are disjoint (even though it may seem from the illustration that some vertices are the same). %\jacob{I am confused about (5). It should be more clear that this is a hybrid division, and the boundary edge should be clearly marked as red.}
%For the definition of hybrid divisions, see \Cref{def:hybrid_division}.
%}
%\end{figure}

As it is useful for the understanding of computation of $r$-divisions in MPC, we now sketch a modification of the algorithm by \citet{chazelle2011online}.
We partition the plane into polygons (which we will call \emph{sectors}) such that all vertices that lie in one polygon/sector belong to the same region. We build a point location data structure for this plane partition and for each polygon in the partition, we store which region its vertices belong to. This together allows us to efficiently determine which region a vertex lies in, just based on its coordinates.
%Given a vertex, we then may determine which region it belongs to by determining which polygon it belongs to, and using that we are storing, for the polygon, which region its vertices are in.

It remains to specify how to find a suitable partition of the plane and how to choose for each polygon, which region its vertices will be in. We will make sure that the partition divides the vertices into ``uniformly-sized chunks"; specifically, we make sure that for some parameter $s$, each sector contains at most some $s$ vertices and is intersected by at most $s$ edges, and at the same time there are $O(n \log (n)/s)$ sectors.
%a uniform bounds on their size and the number of edges that intersect them, as well as there shouldn't be too many of them. Specifically, we require that, for a parameter $s$, there are $O(s)$ vertices and edges intersecting any sector and the number of the sectors is $O(n \log (n)/s)$. We prove that, to get a suitable plane partition, it is sufficient to sample $\Theta(n \log (n)/s)$ edges and vertices of the graph and compute the trapezoid partition of this sample.
Such a partition of the plane can be computed by a modified version of the classic algorithm of \citet{Clarkson1989} -- we sample $O(n \log (n)/s)$ edges and vertices and compute the trapezoid decomposition.

We decide to which region a polygon's vertices belong as follows. We build a graph with vertex set being the set of polygons in the plane partition and an edge between vertices corresponding to two polygons if and only if those polygons are in a certain (quite subtle) sense adjacent\footnote{The formal definition assumes that each polygon contains some subset of its geometric boundary. The relation of two polygons being adjacent then depends on what polygons the boundaries belong to.}. %\jacob{Not true. The actual rule is more complicated.}
This ``sector'' graph is planar and has $O(n \log (n)/s)$ vertices and as such, we can find an $r/s$-division with $O(n\log(n)/r)$ regions and boundary size per region $O(\sqrt{r/s})$. For each region $R$ in this $r/s$-division, we make a region consisting of vertices lying in polygons whose corresponding vertices lie in $R$. The size of each region is then $O(s) \cdot r/s = O(r)$. For each vertex in the boundary of $R$, we get $O(s)$ vertices and edges in the boundary of the new region. Therefore, the size of the 
boundary of any region is $O(\sqrt{s r})$. We also build a point location data structure, which enables us to efficiently tell which regions a vertex lies in.

\subsubsection{MPC algorithms}

In this section, we start by showing a Las Vegas algorithm that computes an $r$-division of the input graph in $O(1)$ rounds, both expected and with high probability. We then use this to design recursive algorithms that solve various problems in expected $O(1)$ rounds with space $O(n^{2/3+\epsilon})$ for any constant $\epsilon > 0$. We explain why we need this much space below (usually, an algorithm with complexity $O(n^{1-\epsilon})$ can be modified by using recursion to get down to complexity $O(n^{\epsilon})$ for any constant $\epsilon>0$).

\paragraph{Computing an $r$-division.} To compute an $r$-division, we make use of the sublinear-time algorithm for $r$-divisions. Specifically, we sample in parallel an appropriate number of edges and vertices, and we send them all to one machine. If we pick the number of sampled edges and vertices to be large enough, we then may execute the sublinear-time algorithm on this one machine. This results in an oracle that, given a vertex, returns the region the vertex belongs to. Moreover, the oracle does not perform any queries and only needs to know the coordinates of the vertex that we are querying. We set the parameters in such a way that the oracle fits onto one machine. We can now distribute the oracle to all machines using the broadcast trees. We can now, locally on each machine, determine which region some given vertex lies in. This allows us to move the edges around in such a way that each region in the $r$-division is stored on consecutive machines while any machine stores only edges belonging to only one region.

\paragraph{Sketch of our algorithm for counting connected components.} We now use our algorithm for computing $r$-divisions to get algorithms for several classic graph problems. For illustration, we sketch the simplest of our algorithms -- an algorithm for counting connected components in a graph. %Note that under the $1$-vs-$2$-conjecture, one needs $\Omega(\log n)$ rounds for this problem without having the embedding.
One of the ideas is that of (recursively) compressing a subgraph, to get a smaller graph, such that a solution on this compressed graph can be used to recover a solution on the bigger original graph. The main difference between our algorithms for different problems lies in how we compress a subgraph. The idea of compressing a subgraph is related to the approach of \citet{Eppstein1996} who use subgraph compression for graph properties (not parameters) in the context of dynamic algorithms for planar graphs.

In order to make the recursion work, we solve a more general problem. Namely, we assume that we are given a graph $G$ with some subset of vertices $\nabla(G)$ that are \emph{marked}. We compute the number of connected components that have empty intersection with $\nabla(G)$ while we ``compress" the other components by computing a graph $G'$ with the following properties: it contains all vertices in $\nabla(G)$, it has $O(\nabla(G))$ vertices, and two vertices in $\nabla(G)$ are connected in $G'$ if and only if they are connected in $G$.
%
%The idea of ``compressing a subgraph'' of a planar graph has been introduced by \citet{Eppstein1996} in the context of dynamic algorithms for planar graphs. We extend this idea and use it to give algorithms for computing graph parameters and not just deciding a property.
%\jakub{re-read this TODO}

This problem can be solved recursively. In the base case, the whole graph fits onto one machine, and we may solve the problem by a simple sequential algorithm. We now describe the recursive case.

%The recursive case works as follows.
We find an $r$-division $\Rr$ for $r = |G|^{1-\epsilon}$ for some $\epsilon>0$. For each region $R \in \Rr$, we mark its boundary in addition to the vertices that are already marked in $G$. We then recursively find the number of connected components that do not intersect $\nabla(R)$ and a graph on $O(\nabla(R))$ vertices that connects two vertices in $\nabla(R)$ iff they are connected in $R$. Taking the union of such graphs over all $R \in \Rr$, we get a graph in which two vertices of $\partial(\Rr) \cup \nabla(G)$ are connected only if they are connected in $G$. We call this graph $H$. If we choose the parameters right, the graph $H$ will fit onto one machine. This allows us to further compress this graph, resulting in a graph with $O(\nabla(G))$ vertices in which two vertices in $\nabla(G)$ are connected iff they are connected in $G$. We compute the number of connected components that have empty intersection with $\nabla(G)$ by adding the respective numbers from the recursive calls with the number of connected components in $H$ that have empty intersection with $\nabla(G)$.

In each successive recursive call, the size of the graph is $|G|^{1-\epsilon}$ for a constant $\epsilon > 0$. The recursion ends if $|G| \leq O(\Ss)$ for $\Ss \geq n^{2/3}$. We thus get that the depth of recursion is $O(1)$. The round complexity is thus  $O(1)$ in expectation and with high probability. This also ensures that the sets $\nabla(G)$ do not get too large, as in each successive call, the size of this set increases, as we will prove, by at most $|\partial(R)| \leq o(n^{1/3}) \leq O(\Ss)$.

\paragraph{Why we need $\Omega(n^{2/3})$ space.} The requirement that $\Ss \geq \Omega(n^{2/3 + \epsilon})$ comes from the following two facts. First, we need that $|\partial(\Rr)| \leq \Ss$ (otherwise we couldn't fit the compressed graph of size $O(|\partial(\Rr)| + |\nabla(G)|)$ onto one machine). 
Second, we need to make sure that we can construct the $r$-division oracle on one machine, which requires the space per machine to be $\Ss \geq \Omega(\frac{n \log n}{s})$.

Thus we need $\Ss\geq \Omega(\max\{ \sqrt{sr}\frac{n\log n}{r} , \frac{n\log n}{s} \})$, which is minimized for $s=\sqrt[3]{r}$. Since we also need $r=O(n^{1-\epsilon})$,
this implies that we need $\Omega(n^{2/3+\epsilon/3}\log n)$ space per machine, and $O(n^{2/3+\epsilon})$ space turns out to be sufficient.

It is usually the case that if one gets an algorithm with $O(n^{1-\epsilon})$ space per machine, it is possible to also get an algorithm with the same asymptotic round complexity and $O(n^{\epsilon})$ space per machine by recursively using the same approach. Why is this not straightforwardly possible in our case? We could compute in such small space a smaller graph on which we would need to recurse. However, it is not clear how to efficiently find an embedding of this smaller graph. Our approach fundamentally depends on having this embedding, and it is thus not clear if this approach is feasible in our situation.

\subsection{Beyond straight-line embeddings}
The above algorithm in fact works whenever the input consists of $x$-monotone curves. We pick an angle $\theta \sim \mathit{Unif}([0,2\pi])$, rotate the whole input by $\theta$ angle around the origin, and subdivide all curves so as to make all the resulting curves are $x$-monotone. This can be done by dividing the curve whenever its tangent at a point is vertical. (We may perform the rotation lazily, whenever we want to use the coordinates of a vertex.) We then prove the following lemma: Consider a ``nice" curve $C$ and mark all points such that the tangent of the curve forms angle $\theta$ with the $x$-axis; then the expected number of marked points is $O(A(C))$ where $A(C)$ is the total absolute curvature of $C$. Applying this lemma to the rotated input, we get that the expected number of points at which the tangent of an edge is vertical, is $O(A(G))$ for $A(G)$ being the total absolute curvature of all the edges.

\subsubsection{Estimating additive Lipschitz parameters}
Suppose we have a graph parameter $\Pi$ that changes by $O(1)$ by removing any single vertex and such that for disjoint graphs $G_1,G_2$, it holds that $\Pi(G_1 \cup G_2) = \Pi(G_1) + \Pi(G_2)$. Given a graph $G$, we want to estimate $\Pi(G)$.
%We will divide the graph into small pieces; denote this graph by $G'$.
We find an $r$-division for appropriately chosen $r$ with $O(\epsilon n)$ vertices and edges in the boundaries. We will want to be able to perform BFS within one region. In order to do this, we remove all vertices with degree $\geq \Theta(1/\epsilon)$. We let $G'$ be a graph obtained from $G$ by removing all boundary vertices, all vertices incident to boundary edges, and all vertices with degree $\geq \Theta(1/\epsilon)$. There are only $O(\epsilon n)$ such vertices and thus $|\Pi(G) - \Pi(G')| \leq O(\epsilon n)$. By additivity, $\Pi(G') = \sum_{H \in cc(G')} \Pi(H)$ where $cc(G')$ denotes the set of connected components of $G'$. We estimate this sum using the Horvitz-Thompson estimator: let $v$ be a random vertex, and let $H_v$ be the connected component of $G'$ that $v$ lies in; $n \Pi(H)/|H_v|$ is then an unbiased estimator of $\Pi(G')$. We take the average of sufficiently many independent copies and use this as our final estimate.

\section{Preliminaries} \label{sec:prelims}
We now give several definitions. All the definitions in this section are commonly used and are not new to this paper. We then give definitions at the beginning of \Cref{sec:sublinear_separators} that are either new to this paper or that we use slightly differently from their common usage.

\subsection{Balanced separators and \texorpdfstring{$r$}{r}-divisions}
The following definition comes (in a slightly different form) from \citet{Frederickson1987}.
%We start by giving several definitions that will be of crucial importance for us.
\begin{definition}
A division $\Rr$ is a system of subsets of $V(G)$ such that each edge's endpoints lie in common $R \in \Rr$, which we call a region. An $r$-division is a division such that $(1)$ each $R \in \Rr$ contains $\leq r$ vertices, $(2)$ it holds $|\Rr| \leq O(n/r)$, and $(3)$ the boundary\footnote{Boundary of $R \in \Rr$, denoted $\partial(R)$, is defined as $R \cap (\bigcup_{R' \in \Rr \setminus \{R\}} R')$} of each $R \in \Rr$ has size $|\partial(R)|\leq O(\sqrt{r})$.
\end{definition}

%\begin{definition}[\citet{Frederickson1987}]
%Let us have a planar graph $G$ with $n$ vertices and a parameter $r$. An $r$-division $\Rr$ of the graph is a division into $O(n/r)$ sets $R \subseteq V(G)$, which we call regions, each with $\leq r$ vertices with the property that for any edge $uv \in E$, it holds that $u$ and $v$ share a region (i.e.~there exists $R\in\Rr$ such that $u\in R$ and $v\in R$).
%A vertex is said to be on a boundary of $R \in \Rr$, denoted $\partial(R)$, if it lies in multiple regions.
%\end{definition}
%In this paper, we assume without loss of generality that for any $v \in \partial R$ for some $R \in \Rr$, it holds that $v$ is adjacent to some vertex outside of $R$. We then equivalently define $\partial(S)$ for $S\subseteq V(G)$ for some graph $G$ as the set of vertices in $S$ that are adjacent to some vertex outside of $S$. For an $r$-division $\Rr$, we let $\partial(\Rr) = \bigcup_{R \in Rr} \partial(R)$.

It is well known that planar graphs have $r$-divisions for any value of $r$ \cite{Lipton1979}. %, there always exists an $r$-division $\Rr$ with $|\partial(R)| \leq O(\sqrt{r})$ for any $R \in \Rr$, and it can be found efficiently. 
Moreover, $r$-divisions can be computed very efficiently:

\begin{fact}[\citet{Klein2013}]
Let $G$ be a planar graph. %There exists an $r$-division with $|\partial(R)| \leq O(\sqrt{r})$ for any $R \in \Rr$. Moreover, such $r$-division can be computed in time $O(n)$.
There is an algorithm that computes an $r$-division in linear time, for any given $r \geq \Omega(1)$.
\end{fact}

%In this paper, we will need a somewhat weaker concept of an $r$-division, which we call the
%\emph{hybrid $r$-division}.
%
%
%\begin{definition}
%Let us have a planar graph with $n$ vertices and a parameter $r$. A hybrid $r$-division of $G$ is a division into $O(n/r)$ regions, each with $\leq r$ vertices. A vertex is said to be on a boundary if it lies in multiple regions. An edge $uv \in E$ is said to be on the boundary if $u$ and $v$ do not share a region.
%\end{definition}

%The notion of boundaries in hybrid $r$-division is of crucial importance. Throughout this paper, we will put restrictions on the size of the boundaries. 

Throughout this paper, we assume that $r$-divisions are stored in a ``dual" representation, where we store for each vertex the set of regions it belongs to, instead of storing for each region the set of vertices it contains.  We assume that the representation implements the following operations in worst case $O(1)$ time: membership, cardinality, iterate through elements (taking $O(1)$ time per element). These operations are implemented for example by Cuckoo hash tables.% \jacob{Changed to assume worst case rather than expected, to ensure we don't run into problems with the "high probability" bound.}

\subsection{Trapezoid Decomposition}
Trapezoid decomposition is a geometric structure, originally designed to solve the point location problem -- one of the fundamental problems in computational geometry. It has also been used in \cite{Clarkson1989} for computing cuttings of the plane (defined in \Cref{sec:sublinear_separators}); this is how we use trapezoid decompositions in this paper. This is usually defined for a set of lines, but we give a more general definition for a set of curves. We will use this in \Cref{sec:beyond_straight_lines} where we deal with non-straight-line embeddings.
\begin{definition}
Given a set $S$ of $x$-monotone curves, for each curve's endpoints add the maximal possible vertical line segment that touches the endpoint and does not intersect any of the curves in $S$. Taking the union of these line segments  with $S$, the trapezoid decomposition of $S$ is defined as the set of connected parts of the plane.
\end{definition}

%In the case when some of the line segments have length $0$ (that is, they are equal to one point), we define the trapezoid partition so as not to contain degenerate segments of width $0$.

A trapezoid decomposition of a set of line segments can be computed in time $O(n \log n)$ using for example the algorithm by \citet{Mulmuley1990}. The algorithm also produces a data structure that allows us to determine which connected part of the plane a point lies in in time $O(\log n)$. This is useful for us as it means that we get algorithms that not only perform few queries, but also have good time complexity.

\subsection{Sublinear-time models of computation}
In the sublinear time regime, we do not have enough time to even read the whole input. This means that it is of crucial importance what queries we may use to access the data. We use a model defined by the following queries: (1) return a vertex picked uniformly at random, (2) pick an edge uniformly at random, (3) given a vertex $v$ and $i \leq \deg(v)$, return the $i$-th neighbor of $v$.

While this model has been considered before, a more standard model allows only for queries $(1),(3)$. For us, this difference is not significant. The reason is that random edge queries can be approximately simulated in time $O(\frac{\alpha n}{m} \cdot \frac{\log ^{3} n}{\varepsilon})$ for $\alpha$ being the arboricity of the input graph using the algorithm of \citet{Eden2019}. For planar graphs, it holds that $\alpha \leq 3$. Conveniently, our algorithms works (after slightly changing some parameters) even when we may sample edges only approximately uniformly, say for $\epsilon = 1/2$. This allows us to use the algorithm from \cite{Eden2019} to implement the random edge queries.

In \Cref{sec:beyond_straight_lines}, we deal with the case of non-straight-line embeddings. In that section, we assume that any neighborhood query also returns the connecting edge. Given an edge, we assume that we may $(a)$ given $\theta[0,2\pi)$ find all points on the edge whose tangent has angle $\theta$ with the $x$-axis, and $(b)$ for a ray we may find the first intersection of the ray and the edge.

\subsection{Massively parallel computation} \label{sec:prelim_mpc}
Massively parallel computation (MPC) is a theoretical model introduced by \citet{Karloff2010} that attempts to capture the popular MapReduce framework for parallel computation. Since its inception, this model has received a lot of attention. We now briefly introduce the model. For a more in-depth treatment, see the lecture notes by \citet{Ghaffari2019}.

In MPC, the computation is performed by $\Mm$ machines in synchronous rounds. Each machine has $\Ss$ words of local memory. At the beginning of the computation, the input -- represented by $N$ words -- is evenly partitioned on $O(N/\Ss)$ machines. Each round looks as follows. At the beginning of a round, each machine is storing the messages that were sent to it in the previous round, in addition to the contents of the memory that it had at the end of the previous round. Each machine then may perform arbitrary computation. Then, each machine stores in a predefined place in its memory all the messages that it wants to send, each message having one recipient. At the end of the computation, each machine stores part of the input. For example, in the case of the minimum spanning tree, we want that the minimum spanning tree is stored on some given $O(n/\Ss)$ machines.
The complexity of MPC algorithms is mainly measured in terms of the round complexity. To simplify the exposition, we assume that one machine can store $O(\Ss)$ words (instead of just $\Ss$).

\subsubsection{Broadcast and convergence-cast trees}
If we have a message of size $k$, we cannot simply broadcast it to all other machines unless $\Ss \geq k \Mm$ as a machine has to store all its outgoing messages at the end of a round. Broadcast trees can be used to accomplish this in $O(1)$ rounds and $k \Mm^\epsilon$ space per machine for some $\epsilon>0$. Similarly, a machine cannot receive a message of size $k$ from all other machines. However, if we have an associative operation $\circ$ and messages $m_1, \cdots, m_{\Mm}$, one may compute $\bigcirc_{i=1}^\Mm m_i$ using the converge-cast trees using per-machine space $k \Mm^{\epsilon}$. See \cite{Ghaffari2019} for explanation.

\subsection{Notation}
We use $V(G),E(G)$ to denote the vertex and edge set of a graph $G$ respectively. In a weighted graph, we denote the weight of an edge $e$ by $w(e)$ and the length (the total weight) of a path $P$ by $w(P)$. We use $cc(G)$ to denote the set of connected components of a graph $G$. For $G$ being an \emph{embedded graph}, we use $cr(G)$ to denote the number of pairwise crossings between edges of the embedding. For $V' \subseteq V(G)$, we use $G[V']$ to denote the subgraph of $G$ induced by $V'$. For a graph $G$ and $V' \subseteq V(G)$, we define the boundary of $V'$, denoted $\partial_G(V')$, as the subset of $V'$ of vertices that have at least one neighbor outside of $V'$. We define $\partial_G^E(V')$ as the set of edges that have one endpoint in $V'$ and the other in $V\setminus V'$. When the graph $G$ is clear from the context, we drop the subscript. We denote the order (number of vertices) of a graph $G$ by $|G|$, and the size (number of edges) by $\|G\|$. We use $d_G(u,v)$ to denote the distance from $u$ to $v$ in the graph $G$. We may drop the subscript if $G$ is clear from the context. We use $D_G$ be the complete graph with $w(uv) = d_G(u,v)$.

\section{MPC algorithm for finding an \texorpdfstring{$r$}{r}-division of a planar graph}

%We now show how our sublinear-time algorithm can be used to design efficient parallel algorithms for various problems. All the algorithms we show are based on a subroutine which we now describe.

%\subsection{Computing an \texorpdfstring{$r$}{r}-pseudodivision}
We now give a parallel algorithm for computing $r$-pseudodivisions. This algorithm will be central to all our algorithms in the rest of this paper. The algorithm is based on a modification of the algorithm of \citet{chazelle2011online} for computing an $r$-pseudodivision of triangulations in sublinear time. We explain that algorithm in detail in \Cref{sec:sublinear_separators}, and now we only give a summary of the result that is sufficient for the use in this section.

\begin{theorem}[Summary of \Cref{thm:sublinear_separators}.] \label[theorem]{thm:summary_of_sublinear_separators}
There is an algorithm that, for a parameter $s$, when given $O(n \log (n) /s)$ uniformly chosen vertices and edges, returns a data structure of size $O(n \log (n) /s)$, that, with high probability, for some hybrid $r$-pseudodivision $\Rr$ with region boundary size of $O(\sqrt{sr})$ and $O(n \log (n)/r)$ regions, answers queries on $\Rr$. That is, given a vertex, it returns the set of regions that it lies in.
\end{theorem}

\begin{algorithm}
$s \leftarrow n \log n/\Ss$\\

Sample $\Theta(\frac{n \log n}{s})$ edges and vertices (for large enough constant in $\Theta$); send them all to machine 1\\
Order the sampled edges and vertices on machine 1 randomly\\
Use \Cref{thm:summary_of_sublinear_separators} on machine
%Execute \Cref{alg:partition_plane} and \Cref{alg:main_algorithm_preprocess_r_div} on machine 1
with parameters
$r,s$%
%; when the algorithm requests a random edge or vertex, use the next in the specified ordering
; This gives us oracle $\Dd$ \label{line:get_an_oracle}\\
Distribute $\Dd$ to all machines\\
On each machine, calculate for each region the number of vertices from that region stored on the machine; aggregate these numbers on machine 1 \label{line:compute_regions_sizes}\\% to get the total size of each region\\
Compute on machine 1 the total size of each region by adding together for each region the numbers computed in the previous step \label{line:agregate_region_sizes}\\
On machine 1, compute a mapping from the set of regions to the machines such that region $R \in \Dd$ is mapped to $\lceil|R|/\Ss\rceil$ consecutive machines \label{line:construct_mapping}\\
For each edge $uv$ such that $\Dd(u) \cap \Dd(v) = \emptyset$, add $v$ to $\Dd(u)$\\
Distribute to all machines the mapping from regions to machines; move the edges accordingly \label{line:shuffle_edges}\\

Verify that region sizes are $\leq r$, number of regions $O(n \log (n)/r)$, and boundaries of all regions have size $O(\sqrt{rn\log n/\Ss})$; restart the algorithm if not (note that the total boundary size can be computed as the difference between $n$ and the total size of regions) \label{line:check_result}\\
\caption{Compute a non-hybrid strong $r$-pseudodivision in MPC} \label{alg:MPC_meta_alg}
\end{algorithm}
%Note that the number of samples used by \Cref{alg:partition_plane} and \Cref{alg:main_algorithm_preprocess_r_div} on \cref{line:get_an_oracle} is $O\Big(\frac{n (\log n + \log \delta^{-1})}{s}\Big)$ in the worst case so there is always enough samples to perform this number of random queries. The algorithm is thus well-defined.

\begin{lemma} \label{lem:mpc_meta_alg}
Let us have $r \geq n \log n / \Ss^{1-\epsilon}$ for $\epsilon>0$ being some constant. After execution of \Cref{alg:MPC_meta_alg}, there exists an $r$-pseudodivision of $G$ with $O(n \log (n)/r)$ regions, boundary of size per region $O(\sqrt{r n \log n /\Ss})$ such that the layout of edges in memory satisfies that
\begin{enumerate}
    \item each machine only stores edges of one region
    \item each region is stored on consecutive machines
    \item $O(n/\Ss)$ machines are used
    \item on the first machine, there is stored for each region the range of machines that it is stored on
\end{enumerate}
The algorithm uses $O(\Ss)$ space per machine, $O(n/\Ss)$ machines, and performs $O(1)$ rounds both expected and with high probability.
\end{lemma}

\begin{proof}
We first focus on correctness of the output. We then argue that the algorithm indeed uses $O(\Ss)$ space per machine and expected $O(1)$ rounds. When the algorithm finished, it means that it has not been restarted on \cref{line:check_result}. Therefore, $\Dd$ represents an $r$-pseudodivision with $O(n \log (n)/r)$ regions and with per-region boundary sizes $O(\sqrt{rn \log n/\Ss})$. The conditions $(1),(2)$ are then satisfied by the mapping computed on \cref{line:construct_mapping}, according to which the edges are mapped to the machines (as we move them accordingly on \cref{line:shuffle_edges}).
Moreover, each region $R \in \Dd$ is stored on $\lceil |R| / \Ss \rceil$ machines, meaning that in total, $O(n/\Ss)$ machines are used, thus satisfying point $(3)$. Moreover, the mapping required in point $(4)$ is computed and stored on machine 1 on \cref{line:construct_mapping}.

By \Cref{thm:sublinear_separators}, $\Dd$ is consistent with an $r$-pseudodivision with $O(n \log(n)/r)$ regions and boundary size per region $\sqrt{s r} = O(\sqrt{r n \log n / \Ss})$ (the equality holds by our choice of $s$) with probability at least $1-O(1/n)$.
We restart the algorithm only when this is not the case, which happens with probability $1-O(\delta) = 1-O(1/n)$; we thus make $O(1)$ restarts in expectation and with high probability.  All lines except from line \cref{line:agregate_region_sizes} clearly take $O(1)$ rounds. \Cref{line:agregate_region_sizes} can be implemented in $O(1)$ rounds using convergence-cast trees (see \Cref{sec:prelim_mpc}).
Thus, each line can be implemented in $O(1)$ rounds, leading to total of $O(1)$ rounds in expectation and with high probability.

It remains to argue that only $O(\Ss)$ space per machine is used. By the assumption $s = n\log n/\Ss$ and $\delta= 1/n$, it holds that the number of sampled edges on line 2 is $O(\frac{n (\log n + \log \delta^{-1})}{s}) = O(\Ss)$. Lines 3,4,5 then do not use asymptotically more space.  Specifically, the space complexity of $\Dd$ is $O(\Ss)$.
%This is also the number of chunks in the partition.
We are assuming $r \geq n \log n / \Ss^{1-\epsilon}$. The number of regions is $O(\frac{n (\log n + \log \delta^{-1})}{r}) \leq O(\Ss^{1-\epsilon})$. \Cref{line:compute_regions_sizes} thus uses $O(\Ss^{1-\epsilon})$ space per machine. \Cref{line:agregate_region_sizes} can be implemented in $O(\Ss)$ space per machine using converge-cast trees. Note that the convergence-cast trees need extra $\Mm^{\epsilon'}$ space in addition to the size of one machine for some $\epsilon' > 0$; this is the reason we need $r \geq n \log n/\Ss^{1-\epsilon}$ for some $\epsilon>0$ and not just $r \geq n \log n/\Ss$. The remaining lines can also be implemented in $O(\Ss^{1-\epsilon})$ space per machine. 
\end{proof}

\section{MPC algorithms on embedded planar graphs}\label{sec:mpc_algorithms}
%\subsection{Applications of \texorpdfstring{$r$}{r}-pseudodivisions in MPC} \label{sec:mpc_algorithms}
We now show how an $r$-pseudodivision can be used to give MPC algorithms for several different problems. As a warm-up, we start with counting connected components. We then show algorithms for bipartition and the minimum spanning tree problem. We then give an algorithm for a problem that generalizes the approximate shortest path problem. We use this to also approximate the diameter/radius of a graph.

The solutions to the mentioned problems follow a common theme. %This approach depends on inspired by the approach that \citet{Eppstein1996} used for dynamic algorithms for planar graphs. \jakub{Is this good?}%\jacob{inspired by, and similar to Galil et al.} %\jakub{Is what I wrote good?} This is inspired by and similar to the approach of \citet{Eppstein1996}. 
Our approach works as follows. We show that when we have a graph $G$, we may replace an induced subgraph of $G$ by some other smaller graph that ``compresses" the induced subgraph, and this does not change the structure of the solutions. This is inspired by and similar to the approach of \citet{Eppstein1996}. We compute these ``compressing graphs" as follows. We find an $|G|^{1-\epsilon}$-pseudodivision $\Rr$ for some $\epsilon>0$, recursively compress each region $R$ into a small graph $H_R$. We then use these graphs $H_R$ to find a graph that compresses the whole graph $G$. We then use this to find a solution to the problem on hand. Specifically, we find an $|G|^{1-\epsilon}$-pseudodivision, and for each region, we compress the rest of the graph except for the region, and then use it to recursively find a solution for that region, passing to the recursive call a compressed version of the rest of the graph.

This does not work in a straightforward way for all above-mentioned problems. For example, for computing (approximately) shortest paths, we need to solve a more general problem in order to make the recursion feasible.

Note that since each machine is only storing vertices of a single region (using \Cref{lem:mpc_meta_alg}), it is straightforward to separately execute the algorithm on each group of machines that store one region. After sending all the graphs $H_R$ for $R \in \Rr$ to one machine, we use this single machine to perform the rest of the computation.

Throughout this section, we use $|G|$ to denote the number of vertices and $\|G\|$ to denote the number of edges of $G$ in a given recursive call. We use $n,m$ to denote the respective values in the original call. In this section, we need to compute $r$-pseudodivisions of various graphs. This can be done using \Cref{alg:MPC_meta_alg}; we do not make this explicit in the rest of this section for the sake of brevity.

The proofs in this section are structured such that we always state a theorem, then we state a stronger claim that implies the theorem. Finally, we give a proof of such claim.

\subsection{Counting connected components}
\begin{theorem} \label[theorem]{thm:MPC_connected_components}
Assume that we are given a graph $G = (V,E)$. For any fixed $\epsilon > 0$, there is an algorithm that returns the number of connected components of $G$; it performs $O(1)$ rounds both in expectation and with high probability, and uses $O(\Ss)$ space per machine and $O(\Mm)$ machines for $\Ss = n^{2/3+\epsilon}$, $\Mm = n^{1/3-\epsilon}$. \end{theorem}

In fact, we prove a stronger claim. We need this in order to be able to give a recursive algorithm. We may recover this theorem by setting $\nabla(G) = \emptyset$ and using $\ell$ as the answer. Note that $O(0) = 0$.

\begin{claim}
Let us have a graph $G$ and a set of marked vertices $\nabla(G)$ with $|\nabla(G)| = k \leq O(\Ss)$. Then, for any fixed $\epsilon > 0$, there is an algorithm that returns a graph $H$ with $V(H) \supseteq \nabla(G)$ on $O(k)$ vertices and a number $\ell$ such that $cc(G) = cc(H) + \ell$ and two vertices in $\nabla(G)$ are connected in $H$ iff they are connected in $G$; it performs $O(1)$ rounds both in expectation and with high probability, and uses $O(\Ss)$ space per machine and $O(\Mm)$ machines for $\Ss = n^{2/3+\epsilon}$, $\Mm = n^{1/3-\epsilon}$.
\end{claim}

\begin{proof}
We give a recursive algorithm. The base case uses a straightforward sequential algorithm, while the recursive case is based on \Cref{alg:MPC_meta_alg}.

\paragraph{Base case.} We start by the base case when $n \leq O(\Ss)$. In this case, the whole input fits onto one machine. For each maximal subset $S \subseteq\nabla(G)$ of vertices that are connected in $G$, we add to $H$ a star whose ray vertices are exactly the vertices $S$. It is easy to see that two vertices from $\nabla(G)$ are connected in $H$ iff they are connected in $G$. We set $\ell$ to be the number of connected components that do not have any vertex marked.
%Then $H$ and $\ell$ satisfy the guarantees from the claim.
In the rest of the proof, we focus on the recursive case (when $n > \Ss$).

\paragraph{Recursive case.} We compute an $r$-pseudodivision $\Rr$ of $G$ for $r = |G|^{1-\epsilon}$. Recursively, we compute for each region $R \in \Rr$ a graph $H_R$ with $\nabla(R) \subseteq V(H)$ for $\nabla(R) = \partial(R) \cup (\nabla(G)\cap R)$ and $\ell_R$ such that two vertices in $\nabla(R)$ are connected in $H_R$ iff they are connected in $G$ and $\ell_R$ is the number of connected components that have empty intersection with $\nabla(R)$. In order to perform the recursive call, we need that $\nabla(R) \leq O(\Ss)$, which holds as we will argue below.

Let $G' = \bigcup_{R \in \Rr} H_R$. We now compute the graph $H$ as follows.
We compute the connected components of $G'$ and for each maximal set $S$ of vertices of $\nabla(G)$ that are connected in $G'$, we add to $H$ a star with its ray vertices being $S$. Let $\ell_{G'}$ be the number of connected components in $G'$ that have empty intersection with $\nabla(G)$. We set $\ell = \ell_{G'} + \sum_{R \in \Rr} \ell_R$.

We now prove correctness. We have from each recursive call that two vertices in $\nabla(R)$ are connected in $H_R$ iff they are connected in $G[R]$. Consider two vertices $u,v \in \nabla(G)$. We argue that they are connected in $H$ if and only if they are connected in $G$.

Consider a $uv$-path $P_H$ in $H$ and let $v_i$ the $i$-th intersection of $P_H$ with $\partial(\Rr) \cup \nabla(G)$. Consider the section of $P_H$ connecting $v_i$ and $v_{i+1}$ for some $i$. This section is contained in $H_{R_i}$ for some $R_i \in \Rr$. It holds $v_i,v_{i+1} \in \nabla(R_i)$, and we thus have from the recursive call that $v_i$ and $v_{i+1}$ is also connected in $G[R_i]$; denote one such connecting path by $P_{H,i}$. Consider the union $\bigcup_i P_{H,i}$. This is clearly a connected graph that contains $u$ and $v$, and these two vertices are thus connected in $G$.

Similarly, consider a $uv$-path $P_G$ in $G$ and let $v_i$ be the $i$-th intersection of $P_G$ with $\partial(\Rr) \cup \nabla(G)$. Consider the section of $P_G$ connecting $v_i$ and $v_{i+1}$ for some $i$. This path section is contained in one region, which we call $R_i$. The vertices $v_i$ and $v_{i+1}$ are thus connected in $G[R_i]$. We thus have from the recursive call that $v_i$ and $v_{i+1}$ are connected in $H_{R_i}$. Denote one such connecting path by $P_{G,i}$. Similarly to the above case, $\bigcup_i P_{G,i}$ is a a connected graph that contains $u$ and $v$, and these two vertices are thus connected in $H$.

The sum $\sum_{R \in \Rr} \ell_R$ counts exactly the connected components of $G$ that have empty intersection with $\bigcup_{R \in \Rr} \nabla(R) = V(G')$ (and thus also with $\nabla(G)$). $\ell_{G'}$ counts the connected components that have non-empty intersection with $G'$ but have empty intersection with $\nabla(G)$. Therefore, $\ell_{G'} + \sum_{R \in \Rr} \ell_R$ counts connected components that have non-empty intersection with $\nabla(G)$, as we desired.

We now argue that the algorithm uses $O(\Ss)$ space per machine, that it performs in expectation $O(1)$ rounds and that in all recursive calls and for all regions $R$, it holds $|\nabla(R)| \leq O(\Ss)$.

In each successive level of recursion, we execute the algorithm on a graph of size $O(|G|^{1-\epsilon})$. We stop when $|G| \leq O(\Ss)$.  Since $\Ss=\Omega(n^{2/3+\epsilon})$ this implies that for any fixed $\epsilon>0$, the depth of recursion is constant.
Since $r = |G|^{1-\epsilon}$, the boundary size per region is $O(\sqrt{r n \log n/\Ss}) \leq o(n^{1/3})$. This means that in each level of recursion, the size of $\nabla(G)$ increases by at most $o(n^{1/3})$.  Since there are $O(1)$ levels of recursion, it holds in all recursive calls that $|\nabla(G)| \leq o(n^{1/3}) \leq O(\Ss)$.

In our algorithm, excluding the call of \Cref{alg:MPC_meta_alg}, we perform $O(1)$ MPC rounds. \Cref{alg:MPC_meta_alg} takes $O(1)$ rounds both expected and with high probability. There are $O(1)$ levels of recursions, resulting in round complexity of $O(1)$ in expectation and with high probability.

It holds $|H| = O(|\nabla(G)|) \leq O(|G'|)$. To prove that the algorithm uses $O(\Ss)$ space per machine, it is thus sufficient to prove that $|G'| \leq O(\Ss)$. It holds $|G'| \leq |\partial(G)| + \sum_{R \in \Rr} |\partial(R)|$. As we have already argued, $|\partial(G)| \leq O(\Ss)$. We have $|\partial(R)| \leq O(\sqrt{rn \log n/\Ss})$ by \Cref{lem:mpc_meta_alg}. The total size of the boundary is thus $O(\sqrt{rn \log n/\Ss}) \cdot O(\frac{n \log n}{r}) \leq O(\Ss)$, as we wanted to show.
\end{proof}

\subsection{Bipartition}
Given a graph $G$ and a set $S \subseteq V(G)$, we let the set of bipartitions $B_G(S)$ be the set of subsets $S'$ of $S$ such that there is a bipartition of $G$ such that all vertices in $S'$ are in one part and all vertices $S \setminus S'$ in the other. Note that if $G$ is not bipartite, $B_G(S)=\emptyset$ for all $S\subseteq V(G)$.  We will need the following lemma about $B_G(S)$. Throughout this section, when we talk of a $2$-coloring, we call the colors white and black. Before showing the algorithm, we will need several lemmas.

The following lemma states, intuitively speaking, that we may replace any connected component $C$ by a simpler graph without changing the set of bipartitions $B_G(S)$.
%{
%\color{purple}
%WE NEED TO WRITE UP A PROOF
%}
\begin{lemma} \label{lem:replace_component_keeps_bipartitions}
Let us have a graph $G$  and a set $S \subseteq V(G)$, and consider one connected component $C$ of $G$. Let us have a $2$-coloring of $C$. Let $G'$ be a graph obtained from $G$ by replacing $C$ with a star whose ray vertices are $C \cap S$ and subdividing edges whose ray vertices are black in the $2$-coloring.
Then $B_G(S) = B_{G'}(S)$.
\end{lemma}
\begin{proof}
  First observe that since $C$ is a connected component of $G$, $B_G(S)$ can be expressed in terms of  $B_{G\setminus C}(S\setminus C)$ and $B_C(S\cap C)$ as
  \[
    B_G(S) = \{ S_1\cup S_2 
    \mid S_1\in B_{G\setminus C}(S\setminus C)
    \wedge S_2\in B_C(S\cap C) \}    
  \]
  Similarly, let $C'=G'\setminus(G\setminus C)$ be the connected component replacing $C$ in $G'$ then
  \[
    B_{G'}(S) = \{ S_1\cup S_2 
    \mid S_1\in B_{G'\setminus C'}(S\setminus C')
    \wedge S_2\in B_{C'}(S\cap C') \}  
  \]
  Now note that $G'\setminus C'=G\setminus C$ and $S\setminus C'=S\setminus C$, so trivially $B_{G'\setminus C'}(S\setminus C')=B_{G\setminus C}(S\setminus C)$.  Also $S\cap C'=S\cap C$ by construction, so we can name that set $S'$ and what remains to be proven is $B_{C'}(S')=B_C(S')$. Since $C$ and $C'$ are both connected and bipartite, each have a unique $2$-coloring (up to naming of colors).
  Let $W\subseteq S'$ be the set of vertices that are white in the $2$-coloring of $C$. Then $B_C(S')=\{W,S'\setminus W\}$, and by construction we also have $B_{C'}(S')=\{W,S'\setminus W\}$.
  \ignore{Is this proof OK?}
\end{proof}

%\jacob{I have replaced $T$ with $S$ in the following lemma, because that makes the subsequent lemma easier to follow. }The following lemma states, intuitively speaking, that when we have a graph $H$ and an induced subgraph $G[S]$ of $G$ for some $S \subseteq V(G)$, such that the set $\partial(S)$ has the same set of bipartitions in these two graphs (that is $B_H(\partial(S)) = B_{G[S]}(\partial(S))$, then as $2$-colorings are concerned, we may replace $G[S]$ by $H$ without ``changing the possible solutions on the rest of the graph''.
%has the same set of bipartitions as $G[S]$ for some $S \subseteq V(G)$
%{
%\color{purple}
%WE NEED TO WRITE UP A PROOF
%}
\begin{lemma} \label{lem:replace_subgraph_extending_coloring}
Let us have a graph $G$ and a subset of its vertices $S \subseteq V(G)$% with $S' \subseteq S$ being vertices of $S$ adjacent to vertices outside of $S$
. Let $G'$ be a graph resulting by replacing $G[S]$ by a graph $H$ with $B_{G[S]}(\partial(S)) = B_H(\partial(H))$. Then for any $2$-coloring $K$ of $G$, it holds that $K |_{(G \setminus S) \cup \partial(S)}$ can be extended to a $2$-coloring of $G'$, and for any $2$-coloring $K'$ of $G'$, it holds that $K' |_{(G \setminus S) \cup \partial(S)}$ can be extended to a $2$-coloring of $G$.
\end{lemma}
\begin{proof}
Observe that $B_{G[S]}(\partial(S))\neq\emptyset$ if and only if $G[S]$ is bipartite, and that $B_H(\partial(H))\neq\emptyset$ if and only if $H$ is bipartite. Thus $B_{G[S]}(\partial(S))=B_H(\partial(H))$ implies that $G[S]$ is bipartite if and only if $H$ is bipartite.  If $G[S]$ and $H$ are not bipartite, $G$ and $G'$ are not bipartite and there is nothing to prove, so suppose they are bipartite.  If $G[S]$ or $H$ has a connected component that does not intersect $\partial(S)=\partial(H)$ then this component is colored independently of the rest of the graph so we may assume that each component of $G[S]$ and $H$ intersects $\partial(S)=\partial(H)$.%  Then the number of connected components in $G[S]$ and $H$ must be the same, and these components must partition $\partial(S)=\partial(H)$ in the same way. \jacob{Here we need the factoring I mentioned in a previous comment.}  If we can show that the lemma holds when $G[S]$ (and thus $H$) is connected, then the result follows easily by induction (just replace one component of $G[S]$ at a time).  So suppose $G[S]$ and $H$ are connected and bipartite. 

Then any $2$-coloring of $K$ of $G$ includes  $2$-colorings of $G[S]$ and $(G\setminus S)\cup\partial(S)$ that agree on $\partial(S)$. Let $W\subseteq\partial(S)$ be the white nodes in this coloring. Then $W\in B_{G[S]}(\partial(S))=B_H(\partial(H))$, so $K|_{\partial(S)}$ extends to a $2$-coloring $K'|_{H}$ of $H$, and then trivially to a $2$-coloring $K'$ of $G'$ by setting $K'|_{G'\setminus H}=K|_{G\setminus S}$.

Conversely, any $2$-coloring $K'$ of $G'$ includes  $2$-colorings of $H$ and $(G\setminus S)\cup\partial(S)$ that agree on $\partial(H)$. Let $W\subseteq\partial(H)$ be the white nodes in this coloring. Then $W\in B_H(\partial(H))=B_{G[S]}(\partial(S))$, so $K'|_{\partial(H)}$ extends to a $2$-coloring $K|_{G[S]}$ of $G[S]$, and then trivially to a $2$-coloring $K$ of $G$ by setting $K|_{G\setminus S}=K'|_{G'\setminus H}$.
\ignore{Is this proof OK?}
\end{proof}

The following lemma states, intuitively speaking, that we may replace an induced subgraph $G[S]$ by some graph $H$ which behaves the same as $G[S]$ with respect to its boundary and some set $T$, and this does not change how $G$ behaves with respect to $T$. This will allow us to construct the graph $H$ in a recursive fashion.
%{
%\color{purple}
%WE NEED TO WRITE UP A PROOF
%}
\begin{lemma} \label{lem:replace_subgraph_keeps_bipartitions}
Let us have a graph $G$, two sets $S,T \subseteq V(G)$, 
and let us have a graph $H$ with %\jacob{$\partial(H)=\partial(S)$ and} \jakub{I actually think this should be $V(H) \supseteq \partial(S) \cup (T \cap S)$. The function $\partial$ is only defined for a subset of a vertex set... (or it would if we didn't forget to include the definition, which I will add).}
$B_H(\partial(S) \cup (T \cap S)) = B_{G[S]}(\partial(S) \cup (T \cap S))$. Let $G'$ be a graph obtained from $G$ by replacing $G[S]$ by $H$. Then $B_{G}(T) = B_{G'}(T)$.
\end{lemma}
\begin{proof}
    Let $X=\partial(S)\cup (T\cap S)$, and observe that by our assumptions, we have $B_H(X)=B_{G[S]}(X)$ and $X\subseteq V[H]$. 
    In particular, we have $\partial(H)=\partial(S)\subseteq X$ so $B_H(\partial(H))=\{Y\cap \partial(H)\mid Y\in B_H(X)\}=\{Y\cap \partial(S)\mid Y\in B_{G[S]}(X)\}=B_{G[S]}(\partial(S))$. Thus by \Cref{lem:replace_subgraph_extending_coloring} $G$ is bipartite if and only if $G'$ is bipartite.
    If $G$ and $G'$ are not bipartite, $B_G(T)=B_{G'}(T)=\emptyset$ and we are done. Suppose therefore that $G$ and $G'$ are bipartite.  Then any connected component in $G$ or $G'$ that does not intersect $T$ has no effect on $B_G(T)$ and $B_{G'}(T)$ so we may assume without loss of generality that every connected component of $G$ and $G'$ intersect $T$.  Then $B_H(X)=B_{G[S]}(X)$ implies that the connected components of $H$ and $G[S]$ partition $X$ the same way, and thus the connected components of $G$ and $G'$ partition $T$ the same way.
    If we can show the result when $G$ and $G'$ are connected, then the general case follows easily by induction on the number of connected components.  Suppose therefore that $G$ and $G'$ are connected and bipartite, and therefore each have a unique $2$-coloring up to the naming of colors.
    
    By \Cref{lem:replace_subgraph_extending_coloring} we have $B_G((G\setminus S)\cup\partial(S))=B_{G'}((G\setminus S)\cup\partial(S))$, and for any $W\in B_G((G\setminus S)\cup\partial(S))$, we can let $K$ and $K'$ be the unique $2$-colorings of $G$ and $G'$ such that every vertex in $W$ is white.
    
    Then in particular $K$ and $K'$ agree on $T\cap((G\setminus S)\cup\partial(S))$ and on $\partial(S)$.
    Since $G$ and $G'$ are both connected, every node in $X$ is connected to a node in $\partial(S)$, and has its color uniquely determined by the coloring of $\partial(S)$. Since $B_H(X)=B_{G[S]}(X)$, $K$ and $K'$ must therefore agree on the color for every vertex in $X$ as well, and in particular they must also agree on $T\cap X$.
    Thus $K$ and $K'$ agree on $T$ and  $B_G(T)=B_{G'}(T)$. 
    \ignore{Is this proof OK?}
\end{proof}

\begin{theorem}
Assume that we are given a graph $G = (V,E)$. For any $\epsilon > 0$, there is an algorithm that computes a $2$-coloring of $G$ or determines that no such $2$-coloring exists; it performs $O(1)$ rounds in expectation and with high probability, and uses $O(\Ss)$ space per machine and $O(\Mm)$ machines for $\Ss = n^{2/3+\epsilon}$, $\Mm = n^{1/3-\epsilon}$.
\end{theorem}
We again prove a stronger claim, in order to enable us to give a recursive algorithm. We can use it to prove the theorem by setting $\nabla(G) = \emptyset$ and using the algorithm $(2)$.
\begin{claim}
Let us have a graph $G$ and a set of marked vertices $\nabla(G)$ with $|\nabla(G)| = k \leq O(\Ss)$. Then, for any fixed $\epsilon > 0$, there is
\begin{enumerate}[(1)]
    \item an algorithm that returns a graph $H$ with $V(H) \supseteq \nabla(G)$ on $O(k)$ vertices such that $B_G(\nabla(G)) = B_H(\nabla(G))$ 
    \item an algorithm that, given a $2$-coloring of $\nabla(G)$
    computes a $2$-coloring of $G$ extending the given $2$-coloring of $\nabla(G)$, or detects that no such $2$-coloring exists
\end{enumerate}
and both algorithms perform $O(1)$ rounds both in expectation and with high probability, and use $O(\Ss)$ space per machine and $O(\Mm)$ machines for $\Ss = n^{2/3+\epsilon}$, $\Mm = n^{1/3-\epsilon}$.
\end{claim}

\begin{proof}
We first give the algorithm $(1)$ and then use it to get algorithm $(2)$.

\paragraph{Base case of algorithm $(1)$.} Assume that $|G| \leq O(\Ss)$. Let us have a connected component $C$ in $G$ with non-empty intersection $S$ with $\nabla(G)$. Let us have a $2$-coloring of $C$ (we call the colors black and white). We add to $H$ a star with its ray vertices being the vertices $S$ and subdivide edges incident to vertices that are black in the $2$-coloring. By \Cref{lem:replace_component_keeps_bipartitions}, the set $B(\nabla(G))$ remains unchanged. We repeat this for all connected components and let $H$ be the resulting graph. We then have $B_H(\nabla(G)) = B_G(\nabla(G))$.
Each connected component $C$ is replaced by a graph on $O(|C \cap \nabla(G)|)$ vertices and the total size of the resulting graph is thus $O(|\nabla(G)|)$. 

\paragraph{Recursive case of algorithm $(1)$.} We compute an $r$-pseudodivision $\Rr$ of $G$ for $r = |G|^{1-\epsilon}$. Let $\nabla(R) = \partial(R) \cup (\nabla(G) \cap R)$. Recursively, we compute for each region $R \in \Rr$ a graph $H_R \supseteq \nabla(R)$ on $O(|\nabla(R)|)$ vertices, such that $B_{H_R}(\nabla(R)) = B_{G[R]}(\nabla(R))$. Note that for this, we need that $|\nabla(R)| = O(\Ss)$; we discuss this below.
Let $G' = \bigcup_{R \in \Rr} H_R$. The graph $G'$ can be obtained by replacing, one by one for each region $R \in \Rr$, the subgraph $G[R]$ by $G[H_R]$. By repeatedly applying \Cref{lem:replace_subgraph_keeps_bipartitions}, we thus have that $B_{G'}(\nabla(G)) = B_G(\nabla(G))$.

We compute a $2$-coloring of $G'$ and for each connected component $C$ of $G'$, we add to $H$ a star with its ray vertices being $C \cap \nabla(G)$ and subdivide edges whose vertex is black in the $2$-coloring. It follows by repeatedly applying \Cref{lem:replace_component_keeps_bipartitions} that $B_{G'}(\nabla(G)) = B_H(\nabla(G))$ and thus $B_G(\nabla(G)) = B_H(\nabla(G))$.

For each connected component $C$, we have added to $H$ a graph on $O(|C \cap \nabla(G)|)$ vertices, and we thus have $|H| \leq O(|\nabla(G)|)$. This proves correctness of the output.

It remains to argue that the algorithm uses $O(\Ss)$ space per machine, that it performs $O(1)$ rounds  in expectation and with high probability, and that in all recursive calls and for all regions $R$, it holds $|\nabla(R)| \leq O(\Ss)$. The argument is exactly the same as that in the proof of \Cref{thm:MPC_connected_components}, and we do not repeat it here.

\paragraph{Base case of algorithm $(2)$.} We perform a breadth first search starting from $\nabla(G)$ (note that these vertices are already colored) and whenever we see a vertex $v$ for the first time, we give it color opposite to the color of the vertex from which we discovered $v$. We then verify that the resulting $2$-coloring is proper. If it is, we return it, otherwise we have detected that no extending $2$-coloring exists.

\paragraph{Recursive case of algorithm $(2)$.} We compute an $r$-pseudodivision $\Rr$ of $G$ for $r = |G|^{1-\epsilon}$. For each $R \in \Rr$, we use algorithm $(1)$ with $\nabla(R) = \partial(R) \cup (\nabla(G) \cap R)$ to get $H_R$ with $B_{H_R}(\nabla(R)) = B_{G[R]}(\nabla(R))$. The boundary size per region is $O(\sqrt{r n \log (n)/M} \leq O(\Ss)$ and we thus have $|\nabla(R)| \leq O(\Ss)$, allowing us to use algorithm (1).

We let $H = \bigcup_{R \in \Rr} H_R$. Again, this graph can be obtained from $G$ by repeatedly replacing $G[R]$ by $H_R$ for all $R \in \Rr$.
By repeated application of \Cref{lem:replace_subgraph_extending_coloring}, we thus have that the given $2$-coloring of $\nabla(G)$ can be extended to a $2$-coloring of $G$ iff it can be extended to a $2$-coloring of $H$, and if it is extendable to some coloring $K$, then $K|_{\partial(\Rr) \cup \nabla(G)}$ is extendable to a $2$-coloring of $H$. Specifically, $H$ can be $2$-colored in a way that extends the given $2$-coloring of $\nabla(G)$. We find such $2$-coloring $K'$ by the same BFS-based algorithm that we used in the base case of this algorithm. Similarly, we get from \Cref{lem:replace_subgraph_extending_coloring} that we may extend $K'|_{\partial(\Rr) \cup \nabla(G)}$ to a $2$-coloring of $G$ and specifically, that we may extend $K' |_{\nabla(R)}$ to a $2$-coloring of $G[R]$. We may thus recursively find such $2$-coloring of $G[R]$. Since these colorings agree on $H$ and there are no edges between the interiors of regions of $\Rr$, these $2$-colorings together give a $2$-coloring of $G$. 

Again, by exactly the the same argument, this algorithm performs $O(1)$ rounds, it always holds $|\nabla(R)| \leq O(\Ss)$, and the graph $H$ fits onto one machine with memory $\Ss$.
\end{proof}

\subsection{Minimum Spanning Tree}
%\jacob{This whole section seems to be reinventing parts of \cite[Section 5]{Eppstein1996}.} \jakub{This is annoying. Is there any excuse why it may make sense for us to write it like this? Maybe our compresison is simpler? Or maybe it is more obvious that ours can be computed in linear space? Or anything like that? Essentially it seems that it is Lemma 22 that essentially appeared in that paper. It this the case? No hope that we prove something they don't? Any chance there is a good excuse why we write it like this? What do we do about this? I think we should also be quite careful not to say that we show something very new if it is not the case.}

Throughout this section, we assume that the edges are linearly ordered. If we only have a quasi-ordering, we may extend it to a linear order. In the case of a weighted graph, this can be done by replacing the weight $w(e)$ by $(w(e), id(e))$ where $id(e)$ is the unique identifier of $e$, and comparing these ordered pairs lexicographically. This assumption guarantees uniqueness of the minimum spanning trees. We denote the minimum spanning tree of $G$ by $T(G)$.

For minimum spanning trees, the concept of compressing a graph in a way that retains the structure of the solutions has been formalized in \cite{Eppstein1996}. For sake of completeness, we now show how to compress a graph for MST below.
\begin{definition} \label{def:compression_mst}
Let us have a graph $G = (V,E)$ and a set $\nabla(G) \subseteq V$.
We say that a graph $H$ with $V(H) \supseteq \nabla(G)$ compresses $G$ with respect to $\nabla(G)$ if for any graph $G' \supseteq G$ with $\partial_{G'}(G) \subseteq \nabla(G)$%
, then for $T$ being the minimum spanning tree of $G'$ and $T'$ the minimum spanning tree of $(G' \setminus G) \cup H$, it holds $T \cap (G'\setminus G) = T' \cap (G' \setminus G)$.
\end{definition}

%\jacob{I added this lemma and proof to help proving \Cref{lem:sufficient_conditions_for_compression}.}
The following gives an alternative definition that is more useful for some of our proofs.
\begin{lemma}\label{lem:simple_condition_for_compression}
%    Let $G=(V,E)$ be a graph, let $\nabla(G)\subseteq V$, and let $H$ be a graph with $V(H)\supseteq \nabla(G)$.  Suppose that for every $u,v\in\nabla(G)$ and $e\in E$, $H$ contains a $u\cdots v$ path of maximum edge weight $\leq w(e)$ if and only $G$ contains such a $u\cdots v$ path. Then $H$ compresses $G$ with respect to $\nabla(G)$.
    Let $G=(V,E)$ be a graph, let $\nabla(G)\subseteq V$, and let $H$ be a graph with $V(H)\supseteq \nabla(G)$. Then the following are equivalent:
    \begin{enumerate}[(i)]
        \item $H$ compresses $G$ with respect to $\nabla(G)$.
        \item For every $u,v\in\nabla(G)$ and weight $w$, $H$ contains a $u\cdots v$ path of maximum edge weight $< w$ if and only $G$ contains such a $u\cdots v$ path.
    \end{enumerate}
\end{lemma}
\begin{proof}
    (i)$\implies$(ii)
    We will prove the contrapositive. So suppose there exists $u,v\in\nabla(G)$ and weight $w$ such that exactly one of $G$ and $H$ contains a $u\cdots v$ path with maximum edge weight $w' < w$. Let $x$ be a new vertex not in either $G$ or $H$, let $ux$ and $xv$ be new edges with weights $w'<w(ux)<w(xv)<w$, and let $G'=G\cup\{ux,xv\}$. Then $xv\in G'\setminus G$ is in exactly one of $T(G')$ and $T((G'\setminus G)\cup H)$, and thus $H$ does not compress $G$ with respect to $\nabla(G)$.
    
    (ii)$\implies$(i)     
    Consider a graph $G'\supseteq G$ with $\partial_{G'}\subseteq\nabla(G)$, and minimum spanning trees $T=T(G)$ and $T'=T((G'\setminus G)\cup H)$.
    We want to show that for every edge $e\in G'\setminus G$, $e\not\in T$ if and only if $e\not\in T'$.
    Let $uv\in G'\setminus G$. Since $T$ is a minimum spanning tree in $G'$, $uv\not\in T$ if and only if there exists a $u\cdots v$ path in $G'$ where every edge has weight $<w(uv)$.
    Similarly, since $T'$ is a minimum spanning tree in $(G'\setminus G)\cup H$, $uv\not\in T'$ if and only if there exists a $u\cdots v$ path in $(G'\setminus G)\cup H$ where every edge has weight $<w(uv)$.
    By our assumption, such a $u\cdots v$ path exists in $G'$ if and only if it exists in $(G'\setminus G)\cup H$.
\end{proof}

We will need a few definitions when trying to compress a graph. We say vertices $v_1, \cdots, v_k \in V$ form an isolated path if $v_1 v_2 \cdots v_k$ is a path, $d(v_2) = 2, \cdots, d(v_{k-1}) = 2$, and $v_2,\ldots,v_{k-1}\in V\setminus\nabla(G)$.
Let $j$ be such that $v_j v_{j+1}$ is the heaviest edge of the isolated path. We say that $G$ results in $G'$ by contracting an isolated path, if $G'$ can be obtained from $G$ by removing $v_1, \cdots, v_k$ and adding an edge $v_1v_k$ whose weight is before $v_{j}v_{j+1}$ in the order. 

We now prove several sufficient conditions under which $G'$ compresses $G$ with respect to a set $\nabla(G)$.
%{
%\color{purple}
%WE NEED TO WRITE UP A PROOF
%}
\begin{lemma} \label{lem:sufficient_conditions_for_compression}
Let us have a graph $G$ and a subset $\nabla(G) \subseteq V(G)$. Then $H$ compresses $G$ with respect to $\nabla(G)$ for:
\begin{enumerate}[(i)]
    \item $H = T(G)$
    \item $H$ being obtained from $G$ by removing a leaf of $G$ that does not lie in $\nabla(G)$
    \item $H$ being obtained from $G$ by contracting an isolated path
    \item there exists $H'$ such that $H$ compresses $H'$ with respect to $\nabla(G)$ and $H'$ compresses $G$ with respect to the same set
    \item for some $S \subseteq V(G)$, $H$ is obtained by replacing $G[S]$ by $H'$ for some $H'$ which compresses $G[S]$ with respect to $\partial(G[S]) \cup (\nabla(G) \cap S)$
\end{enumerate}
\end{lemma}
\begin{proof}
    Let $H$ be obtained from $G$ by one of the processes above. In each case, $H$ satisfies the conditions of \Cref{lem:simple_condition_for_compression}(ii), so it follows that $H$ compresses $G$ with respect to $\nabla(G)$.
 %   \jacob{We could stop the proof here, but the details are below if we want them...} \jakub{I think details are good :-)}
    
    In detail: In case (i), the minimum spanning tree preserves for all $u,v\in V$ the smallest maximum edge weight taken over over all $u\cdots v$ paths, and thus in particular for all $u,v\in\nabla(G)$.
    
    In case (ii), the removed leaf is not on a path between any $u,v\in \nabla(G)$ so removing it changes nothing for the paths we care about.
    
    In case (iii), we explicitly preserve the maximum weight among the replaced edges, again preserving the maximum weight on all paths we care about.
    
    In case (iv), since $H'$ compresses $G$ and $H$ compresses $H'$,
    for any $u,v\in\nabla(G)$ and weight $w$, some $u\cdots v$ path in $G$ has maximum weight $<w$ if and only if some $u\cdots v$ path in $H'$ has maximum weight $<w$ if and only if some $u\cdots v$ path in $H$ has maximum weight $<w$. Thus, $H$ compresses $G$ with respect to $\nabla(G)$.
    
    In case (v) we can even prove something slightly stronger, namely that $H$ compresses $G$ with respect to $X=\nabla(G)\cup\partial(G[S])\supseteq\nabla(G)$. For this, consider $u,v\in X$ and a weight $w$. If some $u\cdots v$ path in $G$ has maximum edge weight $<w$, we can break it into subpaths where each subpath has its endpoints in $X$, and has maximum weight $<w$. Each such subpath is completely contained either in $(G\setminus G[S])\cup \partial(S)$ (and thus exists unchanged in $H$) or in $G[S]$.
    Since $H'$ compresses $G[S]$ with respect to $ \partial(G[S])\cup(\nabla(G)\cap S)$, each subpath of $u\cdots v$ contained in $G[S]$ has a replacement in $H'$ whose max weight is $<w$. Thus, $H$ contains an $u\cdots v$ walk (and therefore a $u\cdots v$ path) whose maximum edge weight is $<w$.  Conversely, if $H$ contains such a path, a completely symmetric argument gives that $G$ does. \ignore{Is this proof OK?}
\end{proof}

We now show a sequential algorithm that compresses an input graph $G$ with respect to a given set $\nabla(G)$.

\begin{lemma} \label{lem:compress_a_graph}
There is an algorithm that, given a graph $G$ and a set $\nabla(G) \subseteq V(G)$, computes $H$ that compresses $G$ with respect to $\nabla(G)$ and such that $|H| = O(|\nabla(G)|)$; it uses $O(|G|)$ space.
\end{lemma}
\begin{proof}
We compute $T(G)$ using a standard algorithm, repeatedly remove any leaf that is not in $\nabla(G)$ and contract all isolated paths of length at least $2$. By \Cref{lem:sufficient_conditions_for_compression} (part (i),(ii),(iii)), the resulting graph compresses $G$ with respect to $\nabla(G)$. The algorithm clearly has linear space complexity. It thus remains to argue that $|H| = O(|\nabla(G)|)$.

Any vertex in $H$ of degree $<3$ must be in $\nabla(G)$, otherwise it would have been either removed or contracted.  Thus, the number of vertices of degree $<3$ in $H$ is $\leq |\nabla(G)|$.  The claim follows as in a tree with $\leq |\nabla(G)|$ nodes of degree $<3$, the total number of vertices is $O(|\nabla(G)|)$. 
%Root the tree arbitrarily and charge each vertex of degree $2$ to the first vertex of higher degree in the direction away from the root. Since any isolated path has length at most $3$, we charge at most $2$ vertices to any single vertex and the number of vertices of degree $2$ is thus at most twice that of vertices of degree at least $3$. In a tree with $|\nabla(G)|$ leaves, the number of vertices of degree $\geq 3$ is $O(|\nabla(G)|)$. The claim follows as all leaves have to lie in $\nabla(G)$ as otherwise we would have removed them.
\end{proof}

\noindent
We are now ready to prove existence of an algorithm for finding the minimum spanning tree.
\begin{theorem}
Let us have a graph $G$. For any fixed $\epsilon > 0$, there is an algorithm that computes a minimum spanning tree of $G$; it performs  $O(1)$ rounds both in expectation and with high probability, and uses $O(\Ss)$ space per machine and $O(\Mm)$ machines for $\Ss = n^{2/3+\epsilon}$, $\Mm = n^{1/3-\epsilon}$.
\end{theorem}

\noindent
It is again more convenient to prove a stronger claim. We can recover the theorem by using part $(2)$ with $\nabla(G) = \emptyset$, and $G'$ being the empty graph.
\begin{claim}
Let us have a graph $G$ and a set of marked vertices $\nabla(G)$ with $|\nabla(G)| = k \leq O(\Ss)$. Then, for any fixed $\epsilon > 0$, there is 
\begin{enumerate}[(1)]
    \item an algorithm that returns a graph $H$ with $V(H) \supseteq \nabla(G)$ on $O(k)$ vertices that compresses $G$ with respect to $\nabla(G)$
    \item an algorithm that, given a graph $G'$ with $G' \cap G \subseteq \nabla(G)$ and $\|G'\| \leq O(\Ss)$, computes $T(G \cup G') \cap G$
\end{enumerate}
and both algorithms perform $O(1)$ rounds both in expectation and with high probability, and use $O(\Ss)$ space per machine and $O(\Mm)$ machines for $\Ss = n^{2/3+\epsilon}$, $\Mm = n^{1/3-\epsilon}$.
\end{claim}
\begin{proof}
\strut
For algorithm $(1)$, in the case when $|G| \leq O(\Ss)$, we may use \Cref{lem:compress_a_graph} to compress $G$ with respect to $\nabla(G)$. 

\paragraph{Recursive case of algorithm $(1)$.} We compute an $r$-pseudodivision $\Rr$ of $G$ for $r = |G|^{1-\epsilon}$. Let $\nabla(R) = \partial(R) \cup (\nabla(G) \cap R)$. We recursively compute $H_R$ for every $R \in \Rr$ that compresses $G[R]$ with respect to $\nabla(R)$. For this we need $|\nabla(R)| \leq O(\Ss)$. Let $G' = \bigcup_{R \in \Rr} H_R$. The graph $G'$ can be obtained from $G$ by repeatedly replacing a subgraph $G[R]$ by the graph $H_R$ for all $R \in \Rr$. It follows from \Cref{lem:sufficient_conditions_for_compression}, part $(v)$, that $G'$ compresses $G$ with respect to $\nabla(G)$. We then use \Cref{lem:compress_a_graph} to get $H$ of order $O(|\nabla(G)|)$ that compresses $G'$, and thus also $G$ (by \Cref{lem:sufficient_conditions_for_compression}, part $(iv)$), with respect to $\nabla(G)$.

It remains to argue that the algorithm uses $O(\Ss)$ space per machine, that it performs in expectation $O(1)$ rounds and that in all recursive calls and for all regions $R$, it holds $|\nabla(R)| \leq O(\Ss)$. The argument is exactly the same as that in the proof of \Cref{thm:MPC_connected_components}, and we do not repeat it here.

\paragraph{Base case of algorithm $(2)$.} When the graph fits onto one machine, we may compute the minimum spanning tree using a standard sequential algorithm.

\paragraph{Recursive case of algorithm $(2)$.} We compute an $r$-pseudodivision $\Rr$ of $G$ for $r = |G|^{1-\epsilon}$. Let $\nabla(R) = \partial(R) \cup (\nabla(G) \cap R)$. For each region $R \in \Rr$, we compute a graph $H_R$ that compresses $G[R]$ with respect to $\nabla(R)$. We let $\bar{H}_R = \bigcup_{R' \in \Rr \setminus \{R\}} H_{R'}$. The graph $G[R] \cup \bar{H}_R \cup G'$ can be obtained from $G \cup G'$ by repeatedly replacing $G[R']$ by $H_{R'}$ for all $R' \in \Rr \setminus \{R\}$. By the repeated application of 
\Cref{def:compression_mst}, we thus have that $T(G[R] \cup \bar{H}_R \cup G') \cap (R\cup G') = T(G \cup V(G')) \cap (R \cup V(G'))$. 
We distribute all the graphs $H_R$ for $R \in \Rr$ to $|\Rr|$ machines. Note that it holds 
\[
%|\bar{H}_R| = 
O(|\partial(\Rr)|) = \sqrt{r|G| \log |G|/\Ss} \cdot \frac{|G| \log |G|}{r} \leq O(n^{2/3 + \epsilon/2} \log^{3/2} n) \leq O(\Ss/n^{\epsilon/3})
\]
and we may thus do so in $O(1)$ rounds by broadcast trees (see \Cref{sec:prelim_mpc}).
We then for all $R \in \Rr$ (on one machine for each such $R$) recursively compute $T(G[R] \cup \bar{H}_R \cup G') \cap R$ and thus $T(G \cup G') \cap R$; we may do this as $|G' \cup \bar{H}_R| \leq |G'| + |\partial(\Rr)| \leq O(\Ss)$. Taking union over all $R \in \Rr$, we have $T(G)$.

The space per machine is $O(\Ss)$ since $|\partial(\Rr)| \leq O(\Ss)$. The expected round complexity is $O(1)$ by exactly the same argument as in the previous algorithms.
\end{proof}

\subsection{Computing spanners and shortest paths with shortcuts}
We now show an algorithm which we will use as a subroutine for finding approximate shortest paths, and radius/diameter. This problem is, in some sense, more general than these two problems. The motivation behind this is that we give a recursive algorithm and we need this more general problem to make the recursion work.

In MPC, by computing a path $P=e_1 \cdots e_k$, we mean that the edges of the path are stored in their order on the machines. That is, edges $e_{k_i+1}, \cdots, e_{k_{i+1}}$ are stored on machine $i$, for some $0 = \ell_1 < \ell_2 < \cdots < \ell_{k'} = k$ for some value of $k'$.

\begin{lemma} \label{lem:mpc_spanners_paths_with_shortcuts}
Assume that we are given a graph $G = (V,E)$ together with a set $\nabla(G) \subseteq V$ of marked vertices such that $|\nabla(G)| \leq O(n^{2/3})$. Let us have some fixed $\epsilon > 0$. There is an algorithm that
\begin{itemize}
    \item computes a graph $A = (\nabla(G),E')$ such that $|E'| \leq O(n^{(1+\epsilon)2/3})$ and for any $s,t \in \nabla(G)$, it holds $d_G(s,t) \leq d_{A}(s,t) \leq O(1)d_G(s,t)$ (that is, $A$ is a $O(1)$-spanner of $D_G[M]$).
    \item Given $s_1,t_1, \cdots, s_k,t_k \in \nabla(G)$, the algorithm computes $k' \leq k$ disjoint paths $P_1, \cdots, P_{k'}$ where $P_j$ is an $s_{i_{j-1}+1}t_{i_j}$-path where $i_0 = 0, i_{k'}=k$, and $i_j < i_{j+1}$ such that $\sum_{i=1}^{k'} w(P_i) \leq O(1) \sum_{i=1}^k d_G(s_i,t_i)$.
\end{itemize}
It performs $O(1)$ rounds both in expectation and with high probability, and uses $O(\Ss)$ space per machine and $O(\Mm)$ machines for $\Ss = n^{2/3+\epsilon}$, $\Mm = n^{1/3-\epsilon}$.
\end{lemma}
\begin{proof}
%Throughout the proof, we use $n$ to refer to the size of the original graph (as opposed to the size of the graph in some specific recursive call).We let $D_{H_1}(H_2)$ for any graphs $H_1,H_2$ be the complete graph on $V(H_2)$ with $w(uv) = d_{H_1}(u,v)$ for all $u,v \in V(H_2)$.
We give a recursive algorithm. It has two parts, each solving one part of the lemma statement.
\paragraph{Base case.} If the graph fits onto one machine, we compute a $(2k-1)$-spanner $A_G$ of $G$ for $k = \lceil 1/\epsilon \rceil$. Such graph can be computed by a standard algorithm (see e.g. \cite{Baswana2007}) in space $O(|\nabla(G)|^{1+1/k}) \leq O(n^{(1+\epsilon)2/3})$ and this is also the space complexity of the spanner.

For each $i \in [k]$, we find a shortest $s_i t_i$-path. Consider the subgraph $Q(G)$ of $G$ created by the union of these paths. For each connected component $K$ in $Q(G)$, we define $i_{\max,K} = \max\{i\,|\, t_i \in K\}$. Let $i_0 = 0$. We set the path $P_j$ for some $j$ as follows. Consider the connected component $K$ containing $s_{i_{j-1}+1}$. We let $P_j$ be the shortest $s_{i_{j-1}+1} t_{i_{\max,K}}$-path and set $i_j = i_{\max,K}$. We end when $i_j = k$.

It clearly holds $i_0 = 0$. The connected component containing $s_{i_{j-1}+1}$ also contains $t_{i_{j-1}+1}$ and thus $i_j \geq i_{j-1}+1$. The paths are disjoint as each is a subset of a different connected component of $Q(G)$. Disjointness implies the first inequality of
\[
\sum_{i=1}^{k'} w(P_i) \leq w(Q(G)) \leq \sum_{i=1}^k d_G(s_i,t_i)
\]
where the second inequality holds because $Q(G)$ is the union of shortest $s_i t_i$-paths.

\paragraph{Recursive algorithm.} We now show the recursive case. It works in two phases, one for each part of the claim. In each successive level of recursion, we execute the algorithm on a graph of size $O(|G|^{1-\epsilon})$. We stop when $|G| \leq n^{2/3+ \epsilon}$. This means that for any fixed $\epsilon$, the depth of recursion is constant. Throughout the proof, we consider $\epsilon$ to be a constant and hide multiplicative $f(\epsilon)$ in the $O$-notation.

\medskip \noindent
\emph{First phase.} We use \Cref{alg:MPC_meta_alg} with $r = |G|^{1-\epsilon}$. By \Cref{lem:mpc_meta_alg}, we get an $r$-pseudodivision $\Rr$ and a memory layout such that each region is stored on consecutive machines and no machine stores edges from multiple regions. This will allow us to run our algorithm recursively, separately on each region, executing it on the group of machines storing the region. For each region $H$, we let the set of marked vertices be $\nabla(H) = (\nabla(G) \cap V(H)) \cup \partial(H)$. It holds that $|\partial(H)| \leq O(\sqrt{r |G| \log |G| / \Ss}) \leq O(n^{2/3})$. This holds on all levels of recursion and the initial set of marked vertices also has size $O(n^{2/3})$ \footnote{With the way we use the caim, it will actually have size $2$, namely it will be equal to $\{s,t\}$. We do not, however, need this.}. The total number of marked vertices of $H$ is thus also $O(n^{2/3})$ as the number of levels of recursion is $O(1)$. 
For each $H \in \Rr$, we recursively 
%We now recursively 
compute a graph $A_H$ on the set of marked vertices $\nabla(H)$ such that for any $u,v \in \nabla(H)$, we have $d_H(u,v) \leq d_{A_H} \leq O(1)d_H(u,v)$.

Consider the graph $A = \bigcup_{H \in \Rr} A_H$ on the vertex set $V(A) = \bigcup_{H \in \Rr} \nabla(H) = \partial(\Rr) \cup \nabla(G)$. The size of this graph is $O(\frac{|G| \log |G|}{r} \cdot r^{(1+\epsilon)2/3
}) \leq O(n^{2/3 + \epsilon})$ and it thus fits onto one machine. We now argue that for any $s,t \in V(A)$, it holds $d_G(s,t) \leq d_A(s,t) \leq O(1)d_G(s,t)$. Consider a shortest $st$-path in $A$ and denote it by $P_A = v_1 v_2 \cdots v_h$ for some integer $h$. For any $i$, we have that $e_i = v_iv_{i+1} \in A_H$ for some $H \in \Rr$. Since $A_H$ is a spanner of $D_H[M(H))]$, it holds that $d_H(v_i,v_{i+1}) \leq w_A(e_i)$. We then have by the triangle inequality that
\[
d_G(s,t) \leq \sum_{i=1}^h d_H(v_i,v_{i+1}) \leq \sum_{i=1}^h w_A(e_i) = d(P_A) = d_A(s,t)
\]
We now prove the other inequality. Consider a shortest $st$-path $P_G$ in $G$. For $h= |P \cap V(A)|$ and $i \in [h]$, let $v_i$ be the $i$-th vertex of $P_G$ that lies in $V(A)$. (Note the different definition of the $v_i$'s in this case from the definition in the above argument for the other inequality). Note that $s,t$ are marked and they therefore lie in $V(A)$. For any $i$, it holds that $v_i,v_{i+1} \in H$ for some $H \in \Rr$. It holds $d_{A_H}(v_i,v_{i+1}) \leq O(1)d_H(v_i,v_{i+1})$ as $A_H$ is guaranteed to be a $O(1)$-spanner of $D_H[M(H)]$ and it holds $d_{D_H(M(H))}(v_i,v_{i+1}) = d_{H}(v_i,v_{i+1})$. By the triangle inequality, we then have that
\[
d_{A_H}(s,t) \leq \sum_{i=1}^k d_{A_H}(v_i,v_{i+1}) \leq O(1) \sum_{i=1}^k d_H(v_i,v_{i+1}) = O(1) d(P_G) = O(1) d_G(s,t)
\]
This concludes proof of the first part.

\medskip \noindent
\emph{Second phase.} For each $i \in [k]$, we compute a shortest $s_i t_i$-path $P_i$ in $A$ (note that $s_i,t_i \in \nabla(G)$ and thus also $s_i,t_i \in A$). Consider a path $P_i$ and let $e$ be the $j$-th edge of path $P_i$. We define $r_i(e) = (i,j)$ and $r_i(e) = \bot$ if $e \not \in P_i$. 
In what follows, we will be comparing these values; we use the lexicographic order for the tuples and we let $\bot$ to be the smallest element. Let $G_P$ be the union of these paths. For each edge $e \in E(G_P)$, let $r_{\max}(e) = \max_{i\in [k]} r_i(e)$.
%and $r_{\min}(e) = \min_{i\in k} r_i(e)$.
Let $i$ be the largest such that $e \in P_i$, we define $t_{\max}(e)$ to be the endpoint of $e$ that comes later in $P_i$. Consider some subgraph $H \in \Rr$. We let $p_0 = (0,0)$. We set values $p_\ell, s_{H,\ell},t_{H,\ell}$ as follows. Consider the smallest pair $(i,j)$ such that $(i,j) \geq p_{\ell-1}$ and such that the $j$-th edge $e$ of $P_i$ lies in $E(A_H)$ (we stop if there is no such tuple). Let $u$ be the endpoint of $e = uv$ that comes first in $P_i$. If $u \neq t_{\max}(e)$, we set $s_{H,\ell} = u$, $t_{H,\ell} = v$, $p_{\ell} = r_{\max}(e)$.

Once we have set the values $s_{H,i},t_{H,i}$ like this, we recursively compute a $O(1)$-approximately shortest $s_{H,i}t_{H,i}$-path in $H$. We let $P_{H,1}, \cdots, P_{H,k_H}$ be the returned paths, where $k_H$ is the number of the paths returned by the recursive call. We let $A_H'$ be the graph on $\nabla(H)$ where for each of these $k_H$ paths, we add an edge between the endpoints of the path with its weight being the length of that path.

Let $Q(G) = \bigcup_{H \in \Rr} A_H'$. We set the values $i_j$ as follows. We let $i_0 = 0$. We consider the connected component of $Q(G)$ that contains $s_{i_{j-1}+1}$ and we let $\ell$ be the greatest such that $t_{\ell}$ lies in this component. We set $i_{j} = \ell$. We let $P'_j$ be the shortest $s_{i_{j-1}+1} t_{i_j}$-path in $Q(G)$. We terminate when $i_j = k$. Note that if this is not the case, we may continue this process. Let $k'$ be the number of such paths that we find.

We now consider one path $P'_j$ (recall that this is a path in the graph $Q(G)$) and we show how to compute a $s_{i_j}t_{i_{j+1}}$-path $P''_j$ in $G$ such that $d_G(P''_j) \leq O(1) d_{Q(G)}(P'_j)$. We will then do this separately for each path $P'_j$. Consider $e_i = uv$, the $i$-th edge in $P'_j$. Since this edge is in $A'_H$ for some $H \in \Rr$, it has to be the case (by the definition of the graph $A'_H$) that we have computed $uv$-path $P_{H,\ell}$ for some value $\ell$ in one of the recursive calls. Let us define $P_{j,i} = P_{H,\ell}$.
We let $P''_j = \bigcup_{i=1}^{|P_j'|} P_{j,i}$. We will show that these paths satisfy the conditions from the statement. First, we show how to rearrange the edges in memory so that each path is stored by consecutive machines in order in which the edges appear in the path.

We compute the prefix sums of the lengths of the paths $P_{j,1}, \cdots, P_{j,|P_j'|}$ as well as the prefix sums of the lengths of the paths $P''_1, \cdots, P''_{k'}$. We call these prefix sums $S_{j,1}, \cdots, S_{j,|P_j'|}$ and $S_1, \cdots, S_{k'}$.
Suppose that the $i$-th edge $e = uv$ of $P_j$ for some $i,j$ lies in $H \in \Rr$ and $u$ comes before $v$ in the order of $P_i$. We send to the machines storing $H$ the value $|(i',j'), s.t. (i',j') \leq (i,j)| = S_{i-1} + S_{i,j-1}$. The position of an edge $e$, denoted $p(e)$, in the desired order in which the edges are to be stored is equal to this value plus the number of edges in $P_{H,i}$ that come before $e$. We then store $e$ on position $p(e) \mod \Ss$ of machine $\lfloor p(e) / \Ss \rfloor$. This ensures that the paths are stored on the machines in the above-mentioned order.

We now argue correctness. The algorithm returns paths $P''_j$ for $j \in [k']$. The path $P''_j$ is an $s_{i_{j-1}+1}t_{i_j}$-path. It clearly holds $i_0 = 0$. Since the connected component of $Q(G)$ that contains $s_{i_{j-1}+1}$, we get that it also contains $t_{i_{j-1}+1}$. We therefore set $i_j \geq i_{j-1}+1$ for all $j \in [k']$. Finally, it holds $i_{k'} = k$; otherwise, consider $i_{k'}$, it has to be the case that $i_{k'} < k$ in which case there is $s_{k'+1}$ in some connected component and the algorithm did not terminate -- a contradiction. It remains to argue that the condition on the total length of the paths $P'_j$ is satisfied.

We will now argue that $\sum_{i=1}^{k'} w(P_i) \leq O(1) \sum_{i=1}^k d_G(s_i,t_i)$. The paths $P_i'$ are disjoint paths in the graph $Q(G)$ (they are disjoint as there is only one path in each connected component). It thus holds
\[
\sum_{i=1}^{k'} w(P'_i) \leq \sum_{e \in Q(G)} w(e)
\]
%For any $i \in [k]$, let $P_{G,i}$ be the shortest $s_i t_i$-path in $G$.
%and consider one value of $i \in [k]$.
%
Let $v_{j,i}$ be the $i$-th vertex of $P'_j$. % It holds $w(P''_j) = \sum_{i=0}^{|P_j|} d_A(v_{j,i},v_{j,i+1})$.
By the inductive hypothesis, we have $w(P_{j,i}) \leq O(1) d_G(v_{j,i},v_{j,i+1})$. Therefore, we have 
{\small
\[
w(P''_j) = \sum_{i=0}^{|P_j'|} w(P_{j,i}) \leq O(1) \sum_{i=0}^{|P_j'|} d_G(v_{j,i},v_{j,i+1}) \leq O(1) \sum_{i=0}^{|P_j'|} d_A(v_{j,i},v_{j,i+1}) = O(1) d_A(s_{i_{j-1}+1},t_{i_j}) = O(1) w(P'_j)
\]
}
%
%It holds $d_G(s_i,t_i) = \sum_{j=1}^\ell d(v_{i_j}, v_{i_{j+1}})$. It holds $d_A(d(v_{i_j}, v_{i_{j+1}}) \leq O(d(v_{i_j}, v_{i_{j+1}}))$. By the triangle inequality, we then have $d_A(s_i,t_i) \leq O(d_G(s_i,t_i))$. It thus holds
%\[
%\sum_{i=1}^k d(P''_i) \leq \sum_{i=1}^k d(s_i,t_i)
%\]
Finally, each edge of $Q(G)$ is in at least one path $P_i$. It thus holds
\[
\sum_{e \in Q(G)} w(e) \leq \sum_{i = 1}^k w(P_i) = \sum_{i=1}^k d_G(s_i,t_i)
\]
Combining these three inequalities, we get
\[
\sum_{i=1}^{k'} w(P_i'') \leq O(1) \sum_{i=1}^{k'} w(P_j') \leq O(1) \sum_{e \in Q(G)} w(e) \leq O(1) \sum_{i=1}^k d_G(s_i,t_i)
\]
as we wanted to show.
\end{proof}

\subsection{Shortest \texorpdfstring{$st$}{st}-path problem}
We now show an algorithm for finding an approximately shortest path between a pair of vertices.
\begin{theorem}
Assume that we are given a graph $G = (V,E)$ and $s,t \in V$. For every fixed $\epsilon > 0$, there is an algorithm that computes an $O(1)$-approximate $st$-path; it performs $O(1)$ rounds both in expectation and with high probability, and uses $O(\Ss)$ space per machine and $O(\Mm)$ machines for $\Ss = n^{2/3+\epsilon}$, $\Mm = n^{1/3-\epsilon}$.
\end{theorem}
\begin{proof}
This is a special case of the shortest paths with shortcuts problem from \Cref{lem:mpc_spanners_paths_with_shortcuts}. Specifically, we let the set $M$ of marked vertices be $\{s,t\}$. We then set $s_1 = s$ and $t_1 = t$. We then have by \Cref{lem:mpc_spanners_paths_with_shortcuts} the desired algorithm.
\end{proof}

\subsection{Diameter/Radius}
Since we give $O(1)$-factor approximation, the two problems are equivalent. We focus on giving approximation algorithm to the diameter problem. This problem is somewhat easier than the shortest path problem in that the answer is one number.
\begin{theorem}
Assume that we are given a graph $G = (V,E)$. For every fixed $\epsilon > 0$, there is an algorithm that computes an $O(1)$-approximation of the diameter and radius of $G$; it performs $O(1)$ rounds both in expectation and with high probability, and uses $O(\Ss)$ space per machine and $O(\Mm)$ machines for $\Ss = n^{2/3+\epsilon}$, $\Mm = n^{1/3-\epsilon}$.
\end{theorem}
\begin{proof}
%The approach is very similar to that for the shortest cycle.
If the graph fits onto one machine, we compute the diameter exactly. We now deal with the case when this the graph does not fit onto one machine.
We use \Cref{alg:MPC_meta_alg} with $r = |G|^{1-\epsilon}$. By \Cref{lem:mpc_meta_alg}, we get an $r$-pseudodivision $\Rr$ and a memory layout such that each region is stored on consecutive machines and no machine stores edges from multiple regions. We recursively compute the approximate diameters of all regions. Since the depth of recursion is $O(1)$ (as each time, we recurse on graphs of size $\leq |G|^{1-\epsilon}$ and stop when $|G| \leq O(n^{2/3+\epsilon})$), this does not increase the asymptotic round complexity of the algorithm.
We use \Cref{lem:mpc_spanners_paths_with_shortcuts} to get a $O(1)$-spanner $A$ of $D_G[\partial(\Rr)]$. We find the diameter of $A$. Let $L$ be the maximum of this diameter and the approximate diameters of the regions of $\Rr$. We return $3L$.

We now argue correctness. We first prove that $diam(G) \leq 3L$. The argument is by induction. When the whole graphs fits into memory, the claim is true as we compute $diam(G)$ exactly; this proves the base case. Consider $u,v \in V(G)$ such that $d(u,v)$ is maximal and consider a shortest $uv$-path $P$. We divide $P$ into three parts -- the first between $u$ and the first vertex of $P$ that lies in $V(A)$, second between the first and last vertex of $P$ that lie in $V(A)$, and third between the last vertex of $P$ that lies in $V(A)$ and $v$. Let $P'$ be the longest of these three parts. It holds $d(P') \geq d(P)/3$. If $P'$ is the middle part of $P$, then we have $L \geq d_A(P') \geq d_G(P)/3$ and thus $3L \geq d_G(P) = diam(G)$. If $P'$ is either the first or the third part, it is the case that $P'$ is contained in one region $H \in \Rr$. Then the recursive call on $H$ returns value at least $d_H(G)$ by the inductive hypothesis. It then holds $L \geq d_H(G) \geq d(P)/3$ and thus $3L \geq d_G(P) = diam(G)$.

It remains to prove that $3L \leq O(1)diam(G)$. The argument is again by induction. The base case is trivial as in the case where the whole graph fits into memory, we compute the diameter exactly. The diameters of the regions of $\Rr$ are no greater than the diameter of $G$. Thus, the recursive calls return answer $\leq O(1) diam(G)$. Since $A$ is a spanner of $D_G[\partial(\Rr)]$, it holds for any $u,v \in \partial(\Rr)$ that $d_A(u,v) \leq O(1) d_G(u,v)$ and thus $diam(A) \leq O(1) diam(G)$. It follows that $L \leq O(1) diam(G)$, as we wanted to prove.
\end{proof}

%\section{Duality and the minimum cut problem}

\section{Finding small separators} \label{sec:sublinear_separators}
We start by giving several definitions that we will need in this section. These definitions are either new (such as (hybrid) $r$-pseudodivisions) or they differ slightly from the way they are sometime used (such as $q$-cuttings). We start by defining two relaxed versions of the $r$-division.
\begin{definition} \label{def:hybrid_division}
A $r$-pseudodivision is a division that satisfies condition $(1)$, i.e. that $|R|\leq r$ for $R\in\Rr$. Bounds on $|\partial(R)|$ and $|\Rr|$ like those in $(2),(3)$ have to then be specified explicitly.

A hybrid division is an arbitrary system of subsets of $V(G)$. A hybrid $r$-pseudodivision is a hybrid division satisfying condition $(1)$.
\end{definition}
The (hybrid) $r$-pseudodivisions have the advantage that they are simpler to compute while they are sufficient for many of the applications that $r$-divisions are commonly used for.

Given a hybrid $r$-division, one may easily get an $O(r)$-division in linear time. This may be done by, for any boundary edge $uv$, adding $v$ to the region of $u$ and symmetrically adding $u$ to the region of $v$. While this may be done by a linear-time sequential algorithm, this cannot be performed in sublinear time, which is the reason we need this notion.

Note that if we have a $r$-pseudodivision with $k$ regions and boundary sizes $\leq \ell$, we may get a $r k / k'$-pseudodivision with $k'$ regions and boundary sizes $\leq \ell k/k'$ by merging regions. It is thus possible to assume that the number of regions is always $\leq n/r$, at the cost of larger $r$ and larger boundaries. For this reason, we do not focus too much on minimizing the number of regions as we can always ensure that, if we need it.

We now give definitions of plane partitions and good cuttings.

\begin{definition}
An open (convex polygonal) plane partition is a partition of $\mathbb{R}^2$ into sectors $S_1, \cdots, S_k$, such that for every $i \in [k]$, the closure of $S_i$ is a polygon and each of its 
edges (excluding endpoints) and vertices belong to $S_j$ for some $j \in [k]$ (that is, one edge (excluding endpoints) or vertex may not lie in multiple $S_j$'s).
\end{definition}

As we only consider convex polygonal plane partitions in this paper, we leave out the ``convex polygonal" and only talk about plane partitions.

We now give a definition of a $q$-cutting.
This notion is central to our paper -- our algorithm for finding $r$-divisions works by first finding a cutting of the input into not too many sets, and then using it to find the division.
The definition we use differs slightly from the commonly used definition in that we require that it is a plane partition and not some plane partition up to a set of measure zero. Our definition also differs in that we treat vertices and edges separately, whereas in the previous work, a point would be considered a degenerate line segment.

\begin{definition} \label{def:good_partition}
We say a plane partition is a $q$-cutting if
%Given a number $s$, we say a plane partition $\mc{P}$ is \emph{good} if
\begin{enumerate}
    \item each sector contains $\leq n/q$ vertices
%    \item the total number of sectors is $O(n \log n/s)$
    \item the number of edges intersecting any sector is $\leq n/q$
%    \item the maximum degree in $G_a(\mc{P})$ is $O(s)$
\end{enumerate}
We say that a $q$-cutting is \emph{good} for a parameter $\delta$ \footnote{We will be using $\delta$ as the failure probability for our algorithms.}, if
\begin{enumerate}
    \setcounter{enumi}{2}
    \item the total number of sectors is $O\big(q (\log n + \log \delta^{-1})\big)$
\end{enumerate}
%We say it is \relax if it satisfies conditions 1 and 3.
\end{definition}

%A trapezoid decomposition induces a plane partition. Specifically, for each face of the decomposition, we have one region and we put arbitrarily each edge and vertex into one of the $S_i$'s. For example, we may pick a vector that is not parallel to any of the line segments, and then put each vertex into the $S_i$ that is incident to the vertex in the direction of this vector, and put each line segment to the same $S_i$ as its endpoint that is further to the left. This allows us answer consistently answer for each edge and vertex which sector it lies in.

\noindent
Given a plane partition, we define its sector graph as follows:

\begin{definition} \label{def:sector_graph}
Let us have a plane partition $\Pp$. We define the sector graph $G_s(\Pp)$ as follows. The vertex set is the set of sectors. There is an edge between vertices corresponding to two sectors $s_1,s_2$ if and only if there exists point $x \in s_1$ such that there exist points in $s_2$ arbitrarily close to $x$ or $x \in s_2$ and such points in $s_1$.
\end{definition}

Note that while it may seem at first glance that the sector graph is equal to the dual graph, this is not the case. The reason is that the $x \in s_1$ may coincide with a vertex $v$ of $\Pp$, resulting in an edge in the sector graph between vertices corresponding to $s_1$ and $s_2$ even if the sectors only meet at $v$ but do not share an edge. The edge set of the sector graph is thus a (potentially strict) superset of the dual graph's edge set.

%\begin{definition}
%Let us have a graph $G$ and a plane partition $\mc{P}$. We define the analysis multigraph of $\mc{P}$ with respect to $G$, denoted $G_a(\mc{P})$, as follows. The vertex set is the set of sectors. For each edge that crosses the shared boundary of $R_1$ and $R_2$, we put an edge into $G_a(\Pp)$ between the corresponding vertices.
%\end{definition}
%%This definition ensures that $G_s(\mc{P})$ is a planar graph.
%Note that this graph may potentially have many more edges then the original graph as any edge of $G$ contributes a number of edges equal to the number of times it crosses between two adjacent sectors of $\mc{P}$.
%Note that the condition for having an edge is symmetric, and the sector graph is thus an undirected graph. %The sector graph has the following important property.

%\begin{definition}
%The size of a sector $S \subseteq \mathbb{R}^2$, denoted $|S|$, with respect to a graph $G$ is defined as $|V(G) \cap S|$
%\end{definition}
%As we always assume that we are given one input graph, we often do not need $G$ explicitly. Note that the size of a sector is not related to its area, diameter, or similar geometric measures.

\subsection{Finding a good \texorpdfstring{$s/n$}{s/n}-cutting}
In this section, we show how to find a good $s/n$-cutting. Our algorithm is a modification of the algorithm from \cite{Clarkson1989}. We show the algorithm so as to be self-contained, but also because it differs slightly from the algorithm of \cite{Clarkson1989} as we have to deal with vertices separately.

We assume that no two vertices have the same $x$-coordinate. This can be satisfied with high probability by rotating (lazily, whenever we query a vertex or edge) the input by an angle picked uniformly at random from a large enough set.

We interpret a trapezoid decomposition as a plane partition, by arbitrarily putting each edge and vertex into one of the sectors. In order to be able to do this consistently in a distributed fashion, we may pick a vector and put each vertex into the sector that is incident in the direction of the vector. We then put each edge into the sector that its left endpoint is in.
%and (interior of an) edge into the sector that is    

%In order to turn a trapezoid decomposition into a plane partition, we 
%When we talk about a trapezoid decomposition representing a plane partition, we assume that each vertex is put into the sector in the direction of a vector, as we described in \Cref{sec:prelims}.

\begin{algorithm}
$E' \leftarrow$ sample $3 n (2\log n + \log \delta^{-1}) / s$ edges\\
$V' \leftarrow$ sample $n (2\log n + \log \delta^{-1}) / s$ vertices\\
Compute the trapezoidal map of $E' \cup V'$ using the algorithm from \cite{Mulmuley1990}, let $\mc{D}$ be the returned data structure for point location for this trapezoidal map.\\
%$\mc{P} \leftarrow$ plane partition induced by $\Tt$
%\Return{$\mc{P}$}
\Return{$\mc{D}$}

\caption{Find a good $s/n$-cutting of the plane} \label{alg:partition_plane}
\end{algorithm}
 
\begin{lemma} \label{lem:main_algorithm_preprocess_r_div}\sloppy
Given $s$ and $\delta > 0$, \Cref{alg:partition_plane} returns a good $s/n$-cutting with probability at least $1-O(\delta)$, represented by a data structure $\mc{D}$ that answers point location queries in time $O(\log (n/s))$. 
It has query complexity $O\Big(\frac{n (\log n + \log \delta^{-1})}{s}\Big)$ and time complexity $O\Big(\frac{n \log (n/s)(\log n + \log \delta^{-1})}{s}\Big)$.
\end{lemma}
\begin{proof}
Let $\Pp$ be the trapezoid map of $E' \cup V'$. We need to prove that $\Pp$ satisfies the conditions 1 - 3 with probability $1-O(\delta)$. The property 3 --- that the number of sectors is $O\big(n (\log n + \log \delta^{-1})/s\big)$ --- amounts to showing that a trapezoidal map has number of faces linear in the number of line segments. This is well known to be true (see e.g. \cite[Lemma 6.2]{geometry_textbook}).

Consider a trapezoid $t$. We show that if $t$ violates property $1$ or $2$, then the probability that $t$ is a trapezoid in $\Pp$ is $\leq O(\delta/n^2)$. Any trapezoid in any possible trapezoidal map is uniquely determined by a pair of edges, namely its lower and upper boundaries. Therefore, there are $O(n^2)$ potential trapezoids. By the union bound, we thus get that none of the violating trapezoids lie in $\Pp$ with probability at least $1-O(\delta)$.

Consider the case that $t$ has size $>s$ --- it violates property $1$. For $t$ to be in $\Pp$, it has to be the case that none of the $>s$ vertices in $t$ have been sampled to $V'$. This happens with probability at most
\[
(1-s/n)^{n (2\log n + \log \delta^{-1})/ s} \leq \exp(- s \log (n^2/\delta)/ s) = \delta/n^2
\]

%Consider the case that the vertex corresponding to $t$ has degree $>s$ in $G_a(\Pp)$
Consider the case that the sector $t$ is intersected by $>s$ edges --- it violates property $2$. %This means that at least $s/2$ edges of the input graph intersect $t$, as any intersecting edge can contribute at most two to the degree of the vertex (because the region is convex, any line can only intersect the boundary twice).
For $t$ to lie in $\Pp$, it has to be the case that none of the $> s$ intersecting edges have been sampled into $E'$. As the number of edges is $\leq 3n-6 < 3n$, a uniformly sampled edge intersects $t$ with probability at least $s/(3n)$. The probability that we sample none of the intersecting edges is thus at most
\[
(1-s/(3 n))^{3 n (2\log n + \log \delta^{-1})/ s} \leq \exp(- s \log (n^2/\delta)/ s) = \delta/n^2
\]

The query complexity is clearly as claimed. As for the time complexity, the dominant cost is computing the trapezoidal map. For input of size $n$, the algorithm for computing trapezoidal decomposition from \cite{Mulmuley1990} runs in time $O(n \log n)$. Therefore, the time complexity is $O\Big(\frac{n (\log n + \log \delta^{-1})}{s} \log \frac{n (\log n + \log \delta^{-1})}{s}\Big) \subseteq O\Big(\frac{n (\log n + \log \delta^{-1}) \log (n/s)}{s}\Big)$.
\end{proof}

\begin{remark}
Note that our analysis does not depend on being able to sample edges (or vertices) exactly uniformly. For example, if we may pick an edge from distribution that gives probability $\geq 1/(2m)$ to every vertex, it is sufficient to let $E'$ be a sample of $6n(\log n + \log \delta^{-1})/s$ and basically the same analysis is sufficient to prove correctness of this algorithm.

This means that if we are not able to sample edges uniformly, we may use the algorithm by \citet{Eden2019} that allows us to approximately sample an edge in time $O(\frac{\alpha n}{m} \cdot \frac{\log ^{3} n}{\varepsilon})$ for $\alpha$ being the arboricity of the graph. This is equal to $O(\log ^{3} n)$ as $\alpha \leq 3$ for planar graphs and we may choose $\epsilon = 1/2$.
\end{remark}

\subsection{Finding an \texorpdfstring{$r$}{r}-pseudodivision, given a good cutting}
\begin{algorithm}
\SetKwInOut{Input}{input}
\Input{Graph $G$ with $n$ vertices and $m$ edges, parameter $s$, good $s/n$-cutting $\mc{P}$, represented by point location data structure $\mc{D}$}

Construct the sector graph $G_s$ of $\mc{P}$\\
$\Rr\leftarrow$ Find a (strong) $r'$-division of $G_s$  for $r' = r/s$

\Return{$(\Rr, \mc{D})$}
\caption{Compute an $r$-pseudodivision oracle} \label{alg:main_algorithm_preprocess_r_div}
\end{algorithm}
The interpretation of the output is the following. For each region in $R \in \Rr$, we have a region $R'$ of $G$ consisting of the vertices that lie in a sector whose vertex in $G_s(\Pp)$ is in $R$. An edge lies in the separator if it intersects one of the sectors that are in the boundary of $\Rr$.

\begin{theorem} \label[theorem]{thm:sublinear_separators}
Assume we are given $(\Rr,\mc{D})$ as returned by \Cref{alg:main_algorithm_preprocess_r_div} with parameters $s$ and $\delta>0$. There exists an algorithm that, given a vertex $v \in G$ returns in expected time $O\big(\log ((n \log \delta^{-1})/s)\big)$ and without making any queries  a set of regions $v$ belongs to. These sets form, with probability at least $1-O(\delta)$, a hybrid $r$-pseudodivision $\Rr'$ with boundary size $O(\sqrt{s r})$ per region. The number of regions is $O\big(n (\log n + \log \delta^{-1})/r\big)$. The set implements in $O(1)$ time the following operations: membership, cardinality, iterate through elements (taking $O(1)$ time per element).% \jacob{are these times actually worst case now?}\jakub{yes. Changed it}
\end{theorem}
\begin{proof}
We condition on $\Pp$ being a good $s/n$-cutting. This happens with probability $1-O(\delta)$ by \Cref{lem:main_algorithm_preprocess_r_div}. We define the division $\Rr'$ as follows. For each region $R \in \Rr$, we define a region $R' \in \Rr'$ consisting of the vertices that lie in sectors whose corresponding vertices are in $R$. 

In the division $\Rr$ of $G'$, each region has size $\leq r' = r/s$. As $\Pp$ is a good $s/n$-cutting, each sector has $\leq s$ vertices. The number of vertices in the region $R' \in \Rr'$ is thus $\leq s r' = r$.

We now argue that $\Rr'$ is a hybrid $r$-pseudodivision with boundary size per region $O(\sqrt{s r})$.
In the division $\Rr$ of $G'$, each region has boundary of size $O(\sqrt{r'}) = O(\sqrt{r/s})$. The vertices that lie in the boundary of $R'$ are exactly the vertices lying in the $O(\sqrt{r/s})$ boundary sectors. Since $\Pp$ is a good cutting, each sector has $O(s)$ vertices, meaning that $R'$ has $O(\sqrt{s r})$ boundary vertices. The boundary edges all intersect at least one of these boundary sectors as otherwise there would have to be adjacent sectors whose vertices in $G_s(\Pp)$ do not share a region; this is a contradiction with $\Rr$ being an $r$-pseudodivision. The number of edges crossing each sector of $\Pp$ is $O(s)$ by the definition of a good cutting.
The number of edges that cross the boundary sectors and thus the number of edges from $R'$ to vertices not in $R'$ is $O(s) \cdot O(\sqrt{r/s}) = O(\sqrt{r s})$.
This proves that the boundary of $R'$ has size $O(s) \cdot O(\sqrt{r/s}) = O(\sqrt{s r})$.

The number of regions in $\Rr$ is $\leq \frac{n (\log n + \log \delta^{-1})/s}{r'} = n (\log n + \delta^{-1}) /r$ by the definition of an $r'$-division, and this is the case even if $\Pp$ is not a good $s/n$-cutting. 

Given a vertex $v$, we may use $\mc{D}$ to find $S_v$ in time $O\big(\log ((n \log \delta^{-1})/s)\big)$. Let $u \in G_s(\Pp)$ be the corresponding vertex of $S_v$. The set of regions $v$ belongs to in $\Rr'$ is the same as the set of regions of $u \in \Rr$. As we mentioned in the preliminaries, we are assuming that we are storing this set in a representation which implements the operations mentioned in the statement in $O(1)$ expected time.
\end{proof}

\section{Beyond straight-line planar embeddings} \label{sec:beyond_straight_lines}

So far, we have only considered planar graphs with planar straight-line embeddings. In this section, we show that both the assumption on having a straight-line embedding, and the assumption of planarity can be relaxed. We show this separately, but it is also possible to combine this, resulting in algorithms for non-planar embeddings with general curves (satisfying some technical assumptions).

\subsection{Beyond straight-line embeddings}
We now show how the assumption of having a straight-line embedding can be weakened significantly. While we focus in this section on ``nice curves'', the same results can also be obtained for piecewise linear embeddings. We discuss this briefly at the end of this section.

Intuitively speaking, the more ``cuved'' the curves that embed the edges are, the more time and queries we will have to have to perform in order to get an $r$-pseudodivision of the same quality. A natural measure of ``how curved a curve is'' is the total absolute curvature. For a curve $C$, we denote this $A(C)$. This value can be interpreted as the ``total angle by which the curve turns''. Specifically, subdividing a curve at its inflection points (under mild technical assumption), $A(C)$ is equal to the sum over all the parts of the angles by which the curve turned in the respective parts.

We show that a simple modification of our algorithm gives an algorithm that is parameterized by $n + A(G)$ instead by $n$, for $A(G) = \sum_{e \in E(G)} A(e)$, under the mild technical assumption that the curve can be subdivided into countably many curves each without inflection points. Throughout this section, we assume that a neighbor query on $v$ gives, apart from the neighbor $u$, also the curve that embeds the edge $uv$. We assume that for an edge $e$ and an angle $\theta \in [0, 2\pi)$, we can find all points $x$ on the curve embedding $e$ such that the tangent to $e$ at $x$ has angle $\theta$ with the $x$-axis.

We now give a few definitions. We use $\kappa(x)$ to denote the curvature at $x \in C$. If the parameterization $\gamma$ is clear from the context, we may abuse the notation and use $\kappa(y)$ for $y \in \mathbb{R}$ instead of $\kappa(\gamma(y))$. For any point $x$ lying on a curve $C$ 
we define $\alpha_C(x)$ to be the angle between the tangent on $C$ at $x$ and the $x$-axis. We drop the subscript when the curve is clear from the context. Let $L(C)$ be the length of the curve $C$.
We use $A(C)$ to denote the total absolute curvature of $C$. The curvature can be equivalently characterized as follows. For proof of this, see for example \cite[Theorem~4.3]{Schlichtkrull}.
\begin{fact} \label[fact]{fact:curvature_is_derivative_of_angle}
Let us have a doubly differentiable curve $C \in \mc{C}^2$ and let $\gamma:[0,L(C)] \rightarrow \mathbb{R}^2$ be its parameterization that has unit speed. It holds $\kappa(y) = \frac{d \alpha(\gamma(y))}{d y}$ at all points of continuity of $\alpha(\gamma(y))$.
\end{fact}

We now prove a lemma that we will use in order to analyse of our modified algorithm.
\begin{lemma} \label{lem:lemma_on_rotated_input}
Let us have $\theta \sim [0,2\pi)$. Assume we are given a doubly differentiable curve $C$
parameterized by its length at unit speed by $\gamma:[0,L(C)] \rightarrow \mathbb{R}^2$ that can be covered up to a set of measure zero by a countable set $\Ii$ of open intervals such that no interval contains an inflection point.
Mark any point $y \in [0,L(C)]$ such that $\alpha(\gamma(y)) = \theta$ and for any $\epsilon > 0$, there exists $0 < \delta < \epsilon$ with $\alpha(\gamma(y-\delta)) \neq \theta$. Let $\ell$ be the number of marked points.
Then $E[\ell] = A(C)/(2\pi)$.
\end{lemma}
\begin{proof}
We may assume that $A(C) < \infty$, as otherwise the claim is $E[\ell] \leq \infty$, which trivially holds. %\jacob{I think we should just make that part of the premise. The case $A(C)=\infty$ is uninteresting anyway.}
For a point $x \in C$ and $y_x = \gamma^{-1}(x)$, we let $\beta(x) = \alpha(\gamma(0)) + \int_0^{y_x} \kappa(y) dy$.
It holds $\alpha(x) = \beta(x) \mod 2\pi$ by Fact~\ref{fact:curvature_is_derivative_of_angle}, second fundamental theorem of calculus and the fact that all points of discontinuity of $\alpha(\gamma(y))$ are jump discontinuities where the function value jumps between $0$ and $2\pi$.
%$\alpha(\gamma(x)) = 0$ and either $\lim_{y' \rightarrow y^-} \alpha(\gamma(x)) = 2\pi$ and $\lim_{y' \rightarrow y^-} \alpha(\gamma(x)) = 2\pi$
%\jakub{I can't find a reference, as people usually ignore the difference between $\alpha$ and $\beta$ but I think it is important for us. Is this sufficient explanation? I don't see how to prove formally without essentialy re-proving the second theorem of calculus.}.
We subdivide the intervals by the smallest possible number of numbers, so as to make sure that for any interval $I = (a,b)$, it holds that $|\beta(\gamma(a)) - \beta(\gamma(b))| < 2\pi$. We denote the set of these subdivided intervals by $\Ii'$.

We subdivided each $I \in \Ii$ into $\lfloor|\beta(\gamma(a)) - \beta(\gamma(b))|/(2\pi)\rfloor+1$ intervals. At the same time, the interval $I$ contributes at least\footnote{In fact, it contributes exactly this much, but we will not need this.} $|\beta(\gamma(y)) - \beta(\gamma(y'))|$ to the total absolute curvature.
Since we are assuming that the total absolute curvature is finite, this made only finitely many subdivisions. $\Ii'$ is thus still countable. Let us have an arbitrary bijection $f:\mathbb{N} \rightarrow \Ii'$. Let $I_i = f(i)$ and let $C_i = \gamma(I_i)$ and let $a$ and $b$ be the endpoints of $I_i$; that is $I_i = (a_i,b_i)$.
Let $\ell_i$ be the number of points $y \in I_i$ such that $\alpha(\gamma(x)) = \theta$. % and $\gamma(y) = x$ for $y \in [y_i, y_{i+1})$.
Since it holds $\ell_i \geq 0$ for all $i$, we have by the monotone convergence theorem that $E[\ell] = \sum_{i=1}^\infty E[\ell_i]$.

On any interval $I_i$, the function $\beta$ is differentiable and (non-strictly) monotonous.
Therefore, $\gamma^{-1}(\beta^{-1}(\theta)) \cap I_i$ is always either empty or an interval $I$. If this set is empty, we do not mark any point. Otherwise, only the value $\inf(I)$ may be marked.
This specifically means that we mark at most one number in $I_i$. This proves the first equality in
\[
E[\ell_i] = P(\theta \in \alpha(C_i)) = |\beta(\gamma(a_i)) - \beta(\gamma(b_{i}))|/2\pi = \frac{1}{2\pi}\left|\int_{y = a_i}^{b_i} \kappa(y) \,dy\right| = \frac{1}{2\pi}\int_{y = a_i}^{b_i} |\kappa(y)| \,dy
\]
The second equality holds because $\alpha(C_i) = \beta(C_i) \mod 2\pi$, the third equality holds by the definition of $\beta$ and basic properties of integrals, and the last equality holds because $\alpha(\gamma(\cdot))$ is on $I_i$ monotonous and the $\kappa(y) = \alpha'(\gamma(y)$ thus has the same sign on the whole interval. We thus have
\[
E[\ell] = \sum_{i=1}^\infty E[\ell_i] = \frac{1}{2\pi}\sum_{i=1}^\infty \int_{y = a_i}^{b_i} |\kappa(y)| \,dy = \frac{1}{2\pi}\int_0^{L(C)} |\kappa(y)| \,dy= A(C)/(2\pi)
\]
\end{proof}

We are now ready to describe our modified algorithm. The only part of our paper where we have used the straight-line assumption is the algorithm for finding a good cutting. In the rest of the paper, we have been using this as a black box. In fact, more can be said about how we used this assumption. We have used that for any subset of edges, each trapezoid in the trapezoidal decomposition has boundary consisting of (parts of) two edges and two of the vertical line segments that are added (where one of them may be degenerate and have length zero). For this, it is sufficient that all the edges are $x$-monotonous. Since we are assuming that we may for any edge $e$ find all points on $e$ that have angle $\pi/2$ with the $x$-axis (i.e.~points where the curve is vertical), we may find all such points and subdivide an edge at these points. Even if we had access to this modified graph, there may be many such points and the graph resulting from subdividing those edge may be arbitrarily large.
%Suppose for now, that we have access to this modified graph. T.

For this reason, we pick $\theta \sim [0,2\pi)$ and rotate the whole input by this angle. A curve is perpendicular to the $x$-axis at a point $x$ iff the angle at that point in the input is $-\theta + \pi/2$ or $-\theta+3\pi/2$. Since these two angles are uniform on $[0,2\pi)$, the expected number of such points on one edge is $O(A(e))$ by \Cref{lem:lemma_on_rotated_input}.% The total number of such points is thus in expectation $O(A(G))$.

Let us have a parameter $s$. We now give an algorithm that returns a \relax cutting (for this parameter $s$) with probability at least $1-O(\delta)$, with expected number of regions $O((n + A(G)) \log(n)/s)$. It has query complexity $O(\frac{n \log n }{s})$.

%In this section, we change the definition of $G_a(\Pp)$ slightly. Consider a sector $t$ and an edge $e$. Orient the edge and let $t_{in}$ be the sector from which $e$ first enters $t$ and similarly $t_{out}$ the sector to which $e$ last leaves from $t$. 
%
%Let $t_1, \cdots, t_k$ for some $k$ be the sectors such that $e$ crosses between them and $t$, in the order in which the edge visits these sectors.
%For a vertex $v$ of $G_a(\Pp)$ corresponding to a sector $t$, we add two edges incident to $v$, namely to the vertices corresponding to the sectors $t_{in}$ and $t_{out}$.
%%
%An edge can only cross from a trapezoid sector to another through the left and right horizontal part of the boundary. This ensures that $t_{in}$ and $t_{out}$ is well-defined\footnote{In the case of non-planar graphs, one can show that $t_{in},t_{out}$ are well-defined by assuming that the embedding has finitely many crossings.}.
%The edge crosses this 
%\jakub{check that this is formally correct, that $k$ is finite, and similar...}

%Namely for a vertex $v$ of $G_a(\Pp)$ corresponding to a sector $t$ and an edge that intersects $t$, we add to $G_a(\Pp)$ two edges, connecting $v$ to the vertices corresponding to the sectors 
%
%we add only two edges incident to an edge for each edge in $G$.

\begin{algorithm}
$\theta \sim \mathit{Unif}([0,2\pi))$\\
Rotate the whole input by angle $\theta$ around the origin; perform the rotation lazily, rotating each vertex or edge when accessed in the rest of the algorithm\\ 
$E' \leftarrow$ sample $6 n (2\log n + \log \delta^{-1}) / s$ edges\\
$E'' = \emptyset$\\
\For{$e \in E'$}{
    $\gamma \leftarrow$ parameterization by length at unit speed of the curve that embeds $e$\\
    $D \leftarrow \emptyset$\\
    For any $y \in [0,L(e)]$ such that $\alpha(\gamma(y))$ is $\pi/2$ or $3\pi/2$, and for any $\epsilon > 0$, there exists $0 < \delta < \epsilon$ with $\alpha(\gamma(y-\delta)) \neq \pi/2,3\pi/2$, add $x$ to $D$\\
    Subdivide the curve embedding $e$ at all points in $D$\\
    Add the resulting curves to $E''$
}
$V' \leftarrow$ sample $n (2\log n + \log \delta^{-1}) / s$ vertices\\
Compute the trapezoidal map of $E'' \cup V'$ using the algorithm from \cite{Mulmuley1990}, let $\mc{D}$ be a data structure for point location for this trapezoidal map\\
\Return{$\mc{D}$}

\caption{Find a \relax $s/n$-cutting of a non-straight-line graph} \label{alg:non_straight_line}
\end{algorithm}

The proof of the following claim is very similar to that of \Cref{lem:main_algorithm_preprocess_r_div}. We therefore just sketch the proof, focusing on the parts where the proof differs.

\begin{lemma} \label{lem:curved_good_partition}
\Cref{alg:partition_plane} returns a \relax cutting with probability at least $1-O(\delta)$ with expected number of sectors $O((n + A(G)) \log (n)/s)$, represented by a data structure $\mc{D}$ that for a given vertex answers which region it lies in. It has query complexity $O(\frac{n (\log n + \log \delta^{-1})}{s})$.
\end{lemma}
\begin{proof}

The query complexity is clearly as claimed. Let $\Pp$ be the trapezoid map of $E'' \cup V'$. We prove that $\Pp$ has the claimed properties
%conditions 1 - 3
with probability $1-O(\delta)$.

We start by arguing
%property 2 --- 
that the number of sectors is in expectation $O((n + A(G)) \log (n)/s)$. The expected total absolute curvature of the edges in $E'$ is $E[A(E')] = |E'| A(G)/n$. Conditioned on $A(E')$, we have $E[|E''| \, |A(E')] = |E'| + A(E')$ by \Cref{lem:lemma_on_rotated_input}. By the law of total expectation, we have $E[|E''|] = |E'|(1+A(G)/n) = O((n + A(G)) \log(n)/s)$. A trapezoid decomposition has linear number of edges in the size of the input (see e.g. \cite[Lemma 6.2]{geometry_textbook}). The complexity of $\Pp$ is thus as claimed.
The rest of the proof is exactly the same as that of \Cref{lem:main_algorithm_preprocess_r_div}.% The only difference is that any edge $G$ contributes $2$ to the degree of any vertex of $G_a(\Pp)$ by the definition of this graph $G_a(\Pp)$ and not because the edge is embedded by a line segment and the trapezoid is convex.
\end{proof}

This allows us to prove what is essentialy a non-straight-line version of \Cref{thm:sublinear_separators}. In this part of the paper, we focus only on the query complexity. The proof is analogous to that of \Cref{thm:sublinear_separators}, so we do not show it in full detail.
\begin{theorem} \label[theorem]{thm:curved_r_divisions}
There is an algorithm that preprocesses the input using $O(n \log (n)/s)$ queries and produces an oracle $\Dd$. Given a vertex, $\Dd$ answers region queries in a way that is consistent with an $r$-pseudodivision with per-region boundary size $O(\sqrt{sr})$ and $O((n+A(G)) \log (n)/r)$ regions.
\end{theorem}
\begin{proof}[Proof sketch.]
We use \Cref{lem:curved_good_partition} to find a \relax cutting with probability $1-O(\delta)$ and expected number of regions $O((n+A(G))\log(n)/s)$. We construct the sector graph $G_s$ of this partition. We find an $r'$-pseudodivision of $G_s$ for $r' = r/s$. Given a vertex $v$, we report that it belongs to the region of the vertex corresponding to the sector that $v$ lies in. The size of any region is then at most $r' \cdot s = r$. The size of boundary of any region is $O(\sqrt{r'} \cdot s) = O(\sqrt{sr})$ by the definition of a \relax cutting. The expected number of regions is $O(\frac{(n+A(G))\log(n)/s}{r'}) = O((n+A(G)) \log (n)/r)$
\end{proof}

Note that if, on average, $A(e) = O(1)$ then the resulting bound is the same as what we get when we have a straight-line embedding (in terms of query complexity). We believe that this would be very often the case in practice and that this new algorithm thus removes to a large extent the limitations of our algorithm for straight-line embeddings.

We may use \Cref{thm:curved_r_divisions} in both the sublinear-time parameter estimation, and in MPC. By changing parameters of the respective algorithms, we get the following two theorems:
\begin{theorem}%\footnote{Set $r = \Theta(\frac{\sqrt{\epsilon  (A+n)^2 \log ^3(n)}}{\sqrt{n}})$ and $s = \Theta(\frac{n^{3/2} \epsilon ^{5/2}}{(A+n) \sqrt{\log (n)}})$ with appropriate constants in $\Theta$.}
Let $\Pi$ be an additive $O(1)$-Lipschitz property. There is an algorithm that returns, with probability at least $2/3$, an additive $\pm \epsilon n$ approximation of $\Pi(G)$. It has query complexity $O\Big(\frac{(n + A(G))\log^{3/2} n}{\sqrt{n} \epsilon^{5/2}}\Big)$.
\end{theorem}
\begin{theorem}
Assume that we are given a graph $G = (V,E)$. For every fixed $\epsilon > 0$, there is an algorithm that $(1)$ counts the number of connected components, $(2)$ finds a bipartition of $G$ or report that $G$ is not bipartite, $(3)$ compute the minimum spanning tree of $G$; compute an $O(1)$-approximation of $(4)$ a shortest $st$-path for given $s,t$, $(5)$ the diameter and radius of $G$; it performs $O(1)$ rounds both in expectation and with high probability, and uses $\Ss = O((n+A(G))^{2/3+\epsilon})$ space per machine and $O((A(G)+n)^{1/3 - \epsilon})$ machines, assuming each curve is stored in $O(1)$ space.
\end{theorem}
%\jakub{Check these statements. I am especially unsure about the second one}

\begin{remark}
The above approach of dividing the curve at some points can be used to work with piecewise-linear edges. Specifically, we may subdivide the edges at the endpoints of the line segments. Each part of an edge then will be $x$-monotonic (as it will be a line segment). $A(G)$ in the above bounds is then replaced by the total number of line segments that embed the edges.
\end{remark}

\subsection{Beyond planar embeddings}
Assuming that we have a $s/n$-cutting with $k$ sectors, then the approach used in \Cref{sec:sublinear_separators} leads without any modifications to an $r$-pseudodivision with the same boundary sizes but $k s/r$ regions instead of $O(n \log(n)/r)$.
It thus remains to show how to find a $s/n$-cutting in the non-planar case that does not have too many sectors. This was again solved in \cite{Clarkson1989}, and we show a modification of their algorithm that deals separately with the vertices.

For a not necessarily planar embedding of a graph, its planarization\footnote{This is different from the notion used in the graph drawing literature where it refers to the problem of removing the smallest possible number of edges to make an input graph planar; see e.g. \cite{Tamassia2013}} is the graph obtained by adding a vertex $w$ for each crossing and replacing each edge $uv$ that passes through a crossing replaced by $w$, by two edges, $uw$ and $wv$. We modify \Cref{alg:partition_plane} to enable it to work with non-planar graphs.
\begin{algorithm}
$E' \leftarrow$ sample $3 n (2\log n + \log \delta^{-1}) / s$ edges\\
$E'' \leftarrow$ planarization of $E'$\\
$V' \leftarrow$ sample $n (2\log n + \log \delta^{-1}) / s$ vertices\\
Compute the trapezoidal map of $E'' \cup V'$ using the algorithm from \cite{Mulmuley1990}, let $\mc{D}$ be the returned data structure for point location for this trapezoidal map.\\
\Return{$\mc{D}$}

\caption{Find a good cutting of the plane} \label{alg:partition_plane_nonplanar}
\end{algorithm}

By exactly the same argument as in the proof of \Cref{thm:sublinear_separators}, we get that the returned trapezoidal map is \relax. It remains to bound the expected number of sectors. Size of the trapezoidal map is linear in the size of $E'' \cup V'$. It holds $|V'| = O(n \log (n)/s)$. It holds $|E''| = |E'| +  cr(E')$. Each of the $cr(G)$ crossings in $E(G)$ is also a crossing in $E'$ with probability
\[
O\Big(\Big(\frac{n \log (n)/s}{m}\Big)^2\Big) = O\big((\log(n)/s)^2\big) 
\]
and we thus have that the expected complexity of the trapezoid decomposition is $O(n \log(n)/s + cr(G) \log^2(n)/s^2)$.
We get the following claim:
\begin{theorem} \label[theorem]{thm:nonplanar}
There is an algorithm that preprocesses the input using $O(n \log (n)/s)$ queries and produces an oracle $(\Rr,\Dd)$.
Then there exists an algorithm that, given a vertex $v \in G$ returns in expected time $O(\log n/s)$ and without making any queries the set of regions $v$ belongs to. These sets form, with probability at least $1-O(\delta)$, a hybrid $r$-pseudodivision $\Rr'$ with boundary size $O(\sqrt{s r})$ per region. The number of regions is $O(n \log (n)/r + cr(G) \log^2(n)/(sr))$. The set implements in expected $O(1)$ time operations: membership, cardinality, iterate through elements (taking $O(1)$ time per element).
\end{theorem}

\section{Open Problems}

%\oldparagraph{\textit{(1.)} Non-planar non-straight-line embeddings.} We have shown how to relax either the assumption of straight-line embedding, or the assumption of planarity. Can we relax both assumptions at the same time? Note that when relaxing the assumption of having straight lines, we have used that the embedding is planar.

\oldparagraph{\textit{(1.)} Property testing with respect to hamming distance.} In this paper, we have shown results for property testing with respect to the normalized vertex edit distance. Is is possible to get similar property testing results with respect to the more common normalized hamming distance? \citet{Czumaj2019} have given a characterization of testable properties in planar graphs (without having an embedding). Proving a separation between the settings with and without the embedding would be very interesting.

\oldparagraph{\textit{(2.)} Random rotation trick.} Our result on non-straight-line embeddings relies on \Cref{lem:lemma_on_rotated_input}. What other problems could this be used for? We also believe the assumption of the curve being divisible into countably many curves without inflection points can be removed. Is this so? It also seems plausible that one could find a generalization of this result to non-differentiable curves. This result would then imply both our result and the result on pievewise-linear edges.

\oldparagraph{\textit{(3.)} Minimum ($st$-)cut MPC algorithm.} Is it possible to compute an (approximately) minimum cut in $O(1)$ MPC rounds and $n^{1-\Omega(1)}$ space per machine? It would be natural to use planar graph duality. However, it is far from clear how to do this since we would need to get a good cutting for the dual of the input graph.

\oldparagraph{\textit{(4.)} Streaming algorithms for embedded planar graphs.} Can some of our techniques be used to give efficient streaming algorithms for embedded planar graphs? Our sublinear-time algorithm can be simulated in $O(1)$ passes. We can then uniformly sample regions from the resulting $r$-pseudodivision. This allows us to approximate Lipschitz additive parameters in the streaming setting. Can one improve upon this simple approach? It would be especially interesting if one could prove a separation between the setting with and without the embedding. The lower bound from \cite{Assadi2021} should be useful in this.

\oldparagraph{\textit{(5.)} MPC algorithms with $n^{2/3-\Omega(1)}$ space per machine.} This seems especially feasible for problems where we recurse on subproblems on planar graphs: counting connected components, bipartition, and minimum spanning tree. This problem may however be rather hard. In that case, is it at least possible to get such algorithm for the euclidean minimum spanning tree problem? Is this possible for the 1-vs-2-cycles problem? Especially for the 1-vs-2-cycles problem, it seems plausible that $O(n^\epsilon)$ space is possible for any $\epsilon > 0$.

\oldparagraph{\textit{(6.)} $(1+\epsilon)$-approximation of shortest paths in MPC.} There exist efficient $1+\epsilon$ approximate distance oracles for planar graphs \cite{Le2021}. Can some of the techniques be used to give better approximation for shortest paths in MPC with $O(n^{1-\Omega(1)})$ space per machine? Are \emph{exact} shortest paths possible?

\section*{Acknowledgement}
The authors are grateful Krzysztof Nowicki whose comments have helped improve our MPC algorithms and who pointed out implications of our work to Euclidean MST.
The authors are grateful to Philip Klein for pointing out some relevant literature.
The authors are also grateful to Mikkel Thorup and Tomáš Gavenčiak for helpful discussions.

\bibliographystyle{plainnat}
\bibliography{literature}

\appendix

%\jacob{Look at these section titles, rearrange?}

\section{Parameter estimation and property testing in planar graphs}
\subsection{Previous work on parameter estimation and property testing in planar graphs}

A series of papers culminating in \cite{Kumar2021} has considered \emph{partition oracles} in planar (and more generally minor-free) graphs with bounded degrees. Partition oracles have been first (implicitly) used by \citet{Benjamini2010}. %Partition oracle is conceptually quite similar to an oracle answering que $r$-division.
A partition oracle divides the vertices into sets of size $O_\epsilon(1)$ such that there are $\leq \epsilon d n$ edges going between different sets (more specifically, given a vertex, the partition oracle answers which set it lies in). Note that one may use a $\Theta(1/\epsilon^2)$-division as the partition, and being able to implement oracles for $r$-divisions would thus imply a partition oracle. We use this idea in \Cref{sec:estimate_Lipschitz_parameters}.

Partition oracles may be used for estimation of additive Lipschitz parameters and property testing\footnote{In fact, even non-additive properties may be estimated using partition oracles, as shown by \citet{Newman2013}.}. For parameter estimation, we first remove the edges that go between different sets. We then estimate by sampling ``how much a connected component contributes on average to the value of the parameter of the whole graph".% In the case of property testing, we try to determine whether most connected components satisfy the property. We also use these ideas with some modifications in our paper. In both cases, we only need $O(d/\epsilon^3)$ partition oracle queries.% This approach is due to \citet{Newman2013}.

The technically challenging part is implementing the partition oracle. Recently, a polynomial-complexity partition oracle has been obtained by \citet{Kumar2021}. Their algorithm has complexity $O(%d^{135}/\epsilon^{15150}+
d^{685}/\epsilon^{30650})$. This result improves upon the previous work on this problem \cite{Hassidim2009,Levi2015}.% The idea of partition oracles was first introduced by \citet{Benjamini2010}.

A characterization of testable properties in general (unbounded degree) planar graphs has been proven by \citet{Czumaj2019}. They show that, intuitively speaking, a property is testable, if and only if testing it reduces to testing $H$-subgraph-freeness for some graphs $H$. %We show that, under the weaker \emph{vertex} edit distance (as opposed to hamming distance), more properties can be tested, such as $k$-colorability.

\subsection{Preliminaries on property testing}
For most computational problems, there are trivial linear lower bounds. In order to be able to achieve sublinear complexity, it is then necessary to relax the problem specification. One natural approach, which works for graph parameters is that of settling for an approximate solution. For graph properties, the notion of property testing is commonly used.

A property on graphs of order $n$ is defined as a subset of all such graphs that is closed under isomorphism. Let us have some metric $d(\cdot, \cdot)$ on graphs. Commonly used metric in property testing is the hamming distance\footnote{number of edges that have to be changed in order to get from one graph to an isomorphic copy of the other.} normalized by the number of edges. In this paper, we instead consider the vertex edit distance normalized by the number of vertices. This is defined, for graphs $G_1,G_2$ as the smallest number of vertices that have to be removed from $G_1$ and $G_2$ in order to get isomorphic subgraphs. 
Note that because of the different normalization, these two metrics are incomparable.

Let us have a property $\Psi$. In property testing, one requires that if $G$ does have property $\Psi$, then the algorithm says that this is indeed the case (with high probability). If, on the other hand, the distance from any graph with property $\Psi$ is at least $\epsilon$, then the algorithm is required to reject the input. The algorithm is allowed to answer arbitrarily otherwise.

We also consider \emph{tolerant property testing}. Tolerant property testing makes the requirements stricter in that the algorithm must also accept the input if the distance to $\Psi$ is $\leq \epsilon_1$ while rejecting if it is $\geq \epsilon_2$. Property testing is then equivalent to tolerant property testing with $\epsilon_1 = 0$ and $\epsilon_2 = \epsilon$. Tolerant property testing can also be equivalently viewed as approximating the distance to $\Psi$.

\subsection{Estimating additive parameters and property testing} 
\label{sec:estimate_Lipschitz_parameters}
In this section, we consider additive $O(1)$-Lipschitz parameters. Recall that a parameter $\Pi$ is $\alpha$-Lipschitz with respect to a metric $d$ if for any $G_1,G_2$, it holds that $|\Pi(G_1) - \Pi(G_2)| \leq \alpha d(G_1,G_2)$. We consider the vertex edit distance metric, which is defined as the smallest number of vertices that needs to be removed from $G_1$ and $G_2$ in order to get two isomorphic subgraphs of $G_1$ and $G_2$. A property is additive if for graphs $G_1,G_2$ on disjoint vertex sets, it holds that $\Pi(G_1 \cup G_2) = \Pi(G_1) + \Pi(G_2)$.

These two conditions are satisfied by a large set of combinatorial optimization problems, including:
minimum vertex cover\footnote{Smallest subset $S$ of vertices such that each edge is incident to some vertex in $S$},
minimum edge cover\footnote{Smallest subset $S$ of edges such that each vertex is incident to some edge in $S$},
maximum independent set\footnote{Largest subset $S$ of vertices such that there are no two adjacent vertices in $S$},
minimum feedback vertex set\footnote{Smallest subset of vertices such that removing these vertices results in a graph with no cycles},
minimum feedback edge set\footnote{Smallest subset of edges such that removing these vertices results in a graph with no cycles},
minimum dominating set\footnote{Smallest subset $S$ of vertices such that each vertex is adjacent to some vertex in $S$},
maximum matching\footnote{Largest subset of edges such that no two of the edges share a vertex},
minimum maximal matching\footnote{Smallest subset of edges such that no two of them share a vertex and adding any other edge would violate this property},
$H$-packing\footnote{Largest set of vertex-disjoint subgraphs of $G$ such that each of these subgraphs is isomorphic to $H$}.

Another class of such problems come from (tolerant) property testing problems with respect to the vertex edit distance. For a property $\Psi$ closed under taking subgraphs, we define a parameter $\Pi_\Psi$ as the smallest number of vertices we need to remove from $G$ to get $G'$ such that $G'$ satisfies $\Psi$. $\Pi$ is then decomposable and $1$-Lipschitz. Such problems include, for example, $k$-coloring or subgraph freeness.

\begin{algorithm}
Compute a hybrid $r$-pseudodivision $\Rr$ using \Cref{alg:main_algorithm_preprocess_r_div} with $r = \sqrt{\epsilon n \log^3 n}$, $s = c \epsilon^{5/2} \sqrt{n /\log n}$ for $c>0$ small enough constant, and $\delta = 1/12$\\
$A \leftarrow 0$\\
\RepTimes{$k = \frac{36 \alpha^2}{\epsilon^2}$}{ \label{line:Lipschitz_loop}
    $v \leftarrow$ sample a vertex uniformly\\ 
    \If{$v$ is an interior vertex of some region $R \in \Rr$ and $d(v) \leq \alpha 18/\epsilon$}{
        $H \leftarrow$ interior vertices of $R$ reachable from $v$ through vertices of degree $\leq \alpha 18/\epsilon$\\
        $A \leftarrow A + \frac{n}{|H|}\Pi(H)$
    }
}

\Return{$A/k$}
\caption{Estimate a $\alpha$-Lipschitz parameter $\Pi$} \label{alg:estimate_Lipschitz_parameter} 
\end{algorithm}

\begin{theorem}
Let $\Pi$ be an additive $\alpha$-Lipschitz property. \Cref{alg:estimate_Lipschitz_parameter} returns, with probability at least $2/3$, an additive $\pm \epsilon n$ approximation of $\Pi(G)$. It has query complexity $O(\sqrt{n \log^3 n}/\epsilon^{5/2})$.
\end{theorem}

\begin{proof}
Throughout this proof, we condition on \Cref{alg:main_algorithm_preprocess_r_div} succeeding.
We begin by analyzing the expectation and variance of $A$. Let $A_i$ be the increment of $A$ in the $i$-th iteration of the loop on \cref{line:Lipschitz_loop}. Let $V'$ be the set of vertices of degree $> \alpha 18/\epsilon$. Let $G'$ be the graph obtained from $G[V \setminus (\partial_V(\Rr) \cup V')]$ by removing the edges in $\partial_E(\Rr)$%\jakub{define $\partial_V,\partial_E$ for hybrid divisions.}
. Note that the algorithm can be seen as in each iteration performing breadth-first-search in $G'$ from $v$ when $v \in V(G')$
and letting $H$ to be the connected component that $v$ lies in, and otherwise letting $H$ be the empty graph, in which case $A_i = 0$. 

\begin{align}
E[A_1] = \sum_{H \in cc(G')} \frac{|H|}{n} \cdot \frac{n}{|H|} \Pi(H) = \sum_{H \in cc(G')} \Pi(H) = \Pi\big(\bigcup_{H \in cc(G')} H\big) = \Pi(G')
\end{align}
We now analyze the variance of $A$. Recall that $\sup X$ is defined as the smallest $x$ such that $P(X > x) = 0$.
\begin{align}
\Var(A) = k \Var(A_1)  \leq k E(A_1^2) \leq k \sup(A_1)^2 \leq k \alpha^2 n^2
\end{align}
where we are using that $\sup(A_1) = \frac{n}{|H|}\sup(\Pi(H)) \leq \alpha n$ where we have $\Pi(H) \leq \alpha |H|$ since the parameter is $\alpha$-Lipschitz and it follows by decomposability that $\Pi((\emptyset, \emptyset)) = 0$. Therefore $\Var(A/k) \leq \alpha^2 n^2 / k \leq \frac{1}{36} \epsilon^2 n^2$.
By the Chebyshev inequality, it now holds that
\[
P\Big(|A/k - \Pi(G')| \geq \epsilon n/3\Big) \leq \frac{\frac{1}{36} \epsilon^2 n^2}{(\epsilon n/3)^2} = 1/4
\]
We now bound $|\Pi(G') - \Pi(G)|$. By \Cref{thm:sublinear_separators}, it holds (recall that we are conditioning on \Cref{alg:main_algorithm_preprocess_r_div} succeeding)
\[
|\partial \Rr| \leq O(\sqrt{s r} \cdot \frac{n \log n}{r}) = O(\sqrt{\tfrac{s}{r}}\cdot n\log n) = O(\frac{c^{1/2}\epsilon}{\log n}\cdot n\log n) = O(c^{1/2} \epsilon n) \leq \epsilon n /(6 \alpha)
\]
where the last inequality holds for $c$ small enough. Since $\Pi$ is $\alpha$-Lipschitz, removing a vertex changes the value by at most $\alpha$ and removing an edge by removing its endpoints changes the value by at most $2\alpha$. Removing $\partial \Rr$ from $G$ thus changes the parameter $\Pi$ by at most $\epsilon n / 3$. Since $G$ is planar, it holds $\|G\| \leq 3n - 6$. There are thus less than $\epsilon n/(3 \alpha)$ vertices with degree $> \alpha 18/\epsilon$. Removing these vertices changes the parameter $\Pi$ by at most $\epsilon n / 3$. Putting this together,
we have $|\Pi(G) - \Pi(G')| \leq \frac{2}{3}\epsilon n$.

By the union bound, the algorithm succeeds and $|A/k - \Pi(G')| \leq \epsilon n/3$ with probability $\geq 3/4$. When the algorithm succeeds, it holds  $|\Pi(G) - \Pi(G')| \leq \frac{2}{3}\epsilon n$. Moreover, the probability that $\Cref{alg:main_algorithm_preprocess_r_div}$ succeeded is at least $1-1/12$. We thus have that with probability at least $1-(1/4 + 1/12) = 2/3$, it holds $|A/k - \Pi(G)| \leq \epsilon n$.

The size of any region in $\Rr$ is $\leq r = \sqrt{\epsilon n \log^3 n}$. In each iteration of the loop in line \cref{line:Lipschitz_loop}, we perform breadth-first-search on a subset of one region. We only explore neighborhoods of vertices of degrees $\leq O(1/\epsilon)$, and thus perform at most $O(1/\epsilon)$ queries per explored vertex. This means that each iteration of the loop uses $\sqrt{n \smash{\log^3} (n)/\epsilon}$ queries. We perform $O(1/\epsilon^2)$ iterations of the loop and the total complexity is thus $O(\sqrt{n \log^3 n}/\epsilon^{5/2})$. The complexity of \Cref{alg:main_algorithm_preprocess_r_div} is $O(n \log (n)/s) = O(\sqrt{n \log^3 n}/\epsilon^{5/2})$.
\end{proof}

\begin{remark}
If we only want an algorithm that returns $0$ when $\Pi(G) = 0$ and it returns $>0$ when $\Pi(G) \geq \epsilon n$, we may get a more efficient algorithm. This is useful for non-tolerant property testing. To get this guarantee, it is sufficient to set $k = \Theta(1/\epsilon)$. We then may set $r = \sqrt{n \log^3 n}, s = c \epsilon^2 \sqrt{n \log n}$ for $c>0$ being a small enough constant. This results in query complexity of $O(\sqrt{n \log^3 n}/\epsilon^{2})$ instead of $O(\sqrt{n \log^3 n}/\epsilon^{5/2})$.
\end{remark}

\end{document}